\PassOptionsToPackage{table,xcdraw,dvipsnames}{xcolor}
\documentclass[sigconf]{aamas}

\usepackage{balance} 
\usepackage{graphicx} 
\usepackage{hyperref}
\usepackage{url}
\usepackage{amsfonts}
\usepackage{amsmath}
\usepackage{amsthm}
\usepackage{array}
\usepackage{multirow}
\usepackage{tabularx}
\usepackage{booktabs}
\usepackage{makecell}

\usepackage{multirow}
\usepackage{subcaption}
\usepackage[ruled]{algorithm}
\usepackage{caption}
\usepackage[noend]{algpseudocode}
\usepackage{tcolorbox}
\usepackage{float}
\usepackage{chngcntr} 

\makeatletter
\gdef\@copyrightpermission{
  \begin{minipage}{0.2\columnwidth}
   \href{https://creativecommons.org/licenses/by/4.0/}{\includegraphics[width=0.90\textwidth]{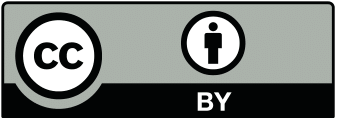}}
  \end{minipage}\hfill
  \begin{minipage}{0.8\columnwidth}
   \href{https://creativecommons.org/licenses/by/4.0/}{This work is licensed under a Creative Commons Attribution International 4.0 License.}
  \end{minipage}
  \vspace{5pt}
}
\makeatother

\setcopyright{ifaamas}
\acmConference[AAMAS '25]{Proc.\@ of the 24th International Conference
on Autonomous Agents and Multiagent Systems (AAMAS 2025)}{May 19 -- 23, 2025}
{Detroit, Michigan, USA}{Y.~Vorobeychik, S.~Das, A.~Nowé  (eds.)}
\copyrightyear{2025}
\acmYear{2025}
\acmDOI{}
\acmPrice{}
\acmISBN{}

\acmSubmissionID{<<OpenReview submission id>>}

\title[AAMAS-2025 Formatting Instructions]{More Efficient Sybil Detection Mechanisms Leveraging Resistance of Users to Attack Requests}

\titlenote{This is the full version of the paper accepted in The 24th International Conference on Autonomous Agents and Multiagent Systems (AAMAS 2025) with appendices.}

\author{Ali Safarpoor Dehkordi}
\affiliation{
  \institution{Australian National University}
  \city{Canberra}
  \country{Australia}}
\email{ali.safarpoordehkordi@anu.edu.au}

\author{Ahad N. Zehmakan}
\affiliation{
  \institution{Australian National University}
  \city{Canberra}
  \country{Australia}}
\email{ahadn.zehmakan@anu.edu.au}

\begin{abstract}
We investigate the problem of
sybil (fake account) detection in social networks from a graph algorithms perspective, where graph structural information is used to classify users as sybil and benign. We introduce the novel notion of user resistance to attack requests (friendship requests from sybil accounts). Building on this notion, we propose a synthetic graph data generation framework that supports various attack strategies. We then study the optimization problem where we are allowed to reveal the resistance of a subset of users with the aim to maximize the number of users which are discovered to be benign and the number of potential attack edges (connections from a sybil to a benign user). Furthermore, we devise efficient algorithms for this problem and investigate their theoretical guarantees.
Finally, through a large set of experiments, we demonstrate that our proposed algorithms improve detection performance notably when applied as a preprocessing step for different sybil detection algorithms.
The code and data used in this work are publicly available on GitHub.\footnote{GitHub repository: \url{https://github.com/aSafarpoor/AAMAS2025-Paper/tree/main}}
\end{abstract}



\definecolor{lightblue}{rgb}{0.68, 0.85, 0.9}
\newtheorem{theorem}{Theorem}[section] 
\newtheorem{definition}{Definition}[section] 
\newtheorem{observation}{Observation}[section] 

\begin{document}


\pagestyle{fancy}
\fancyhead{}


\maketitle 


\section{Introduction}
\textbf{Motivation.} Online social networks (OSNs) like Facebook, Instagram, and Twitter (now X) have become essential parts of our lives. They are places where we connect with friends and family, get the latest news, plan events, and even conduct business. Thus, they have revolutionized many fundamental aspects of our lives.

However, the rapid growth of OSNs has raised privacy and security concerns as malicious entities exploit them for identity theft and spreading harmful content. A key element is sybils (also referred to as “fake”, “malicious”, and “fraudulent” accounts), used for spamming, malware, and phishing, causing financial and reputational damage~\cite{ramalingam2018fake,adewole2017malicious}. Thus, detecting sybils is crucial for a safer online environment.

\textbf{Early Detection Mechanisms.} 
Detecting sybils can be conducted in different stages of the formation. 
The first stage is during registration, 
by using data like IP, Captcha, and location~\cite{yuan2019detecting,liang2021unveiling},
The next stage to detect sybils that pass through the registration is based on their first actions after joining the network, such as the ratio of accepted friend requests, the randomness of the requests sent, or the total number of requests, cf.~\cite{breuer2020friend, breuer2023preemptive}.

\textbf{Late Detection Mechanisms.} While many sybils can be caught using early detection mechanisms, more sophisticated attackers can bypass these by mimicking benign users. Numerous studies have focused on identifying these hard-to-find accounts. Some research analyzes the content provided by users to detect fake or malicious activities, cf. \cite{uppada2022novel}.
Other studies~\cite{ercsahin2017twitter,le2020application} examine user profiles, considering factors such as profile pictures, the number of followers and followings, or the volume of activities like the number of posts. While these methods are more in-depth than the early detection mechanisms, they may still fail to identify sybils that sophisticatedly mimic benign users' behavior.

\textbf{Graph Based Detection Mechanisms.} While sybils have a lot of power over manipulating information such as their IP, username, number of posts, and profile pictures, which are used by the aforementioned methods, they have very little control over their position with respect to the overall network structure. Thus, a natural question arises: can network structural information be used to detect sophisticated sybil attacks? More precisely, suppose we are given a graph $G$, corresponding to a social network, where some nodes (users) are labeled sybil or benign. The goal is to label the remaining nodes as accurately as possible. Prior work has leveraged various methods such as random walk~\cite{jia2017random}, belief propagation~\cite{wang2017sybilscar}, and graph neural networks (GNN)~\cite{yang2020fake}. 

\textbf{Shortcomings of Exiting Methods.} The existing algorithms~\cite{yu2006sybilguard,wang2017gang,lu2023sybilhp} are usually designed under the \textit{homophily} assumption and very limited number of attack edges, that is, most edges (links) in the graph are between nodes of the same type (sybil-to-sybil or benign-to-benign) with significantly fewer edges in between (sybil-to-benign or benign-to-sybil). Due to the lack of real-world data, these algorithms are usually tested on synthetic graph models, which are tailored to possess such homophily properties. However, based on current research~\cite{xu2024revisiting},
we know that the homophily property is not satisfied in many real-world examples.

\textbf{Missing Piece of Puzzle.} One key aspect that steers the formation of connection in a network is how benigns react to friendship requests from sybils. We introduce the notion of user resistance, where a resistant user rejects such requests while a non-resistant accepts them. (For simplicity, we suppose a user is either resistant or non-resistant, but our results can be extended to the setup where nodes have various degrees of resistance.) The notion of resistance permits us to design a more dynamic graph generation framework where, unlike prior static models~\cite{yu2008sybillimit,gong2014sybilbelief}, the outcome graph is a function of the attack strategy and resistance of users rather than a simplified pre-assumption such as homophily. We then devise effective and efficient mechanisms to leverage resistance information for finding potential attack edges and benigns, which serve as helpful preprocessing steps for sybil detection algorithms.

\textbf{Attack Models.}
It is commonly assumed, cf.~\cite{lu2023sybilhp}, that an attacker connects its created sybil accounts together in a fashion that it mimics benigns, for example, by copying the connections between a set of benigns with the same size.

The vital question is how the connections between sybils and the rest of the network are formed. Prior models~\cite{wang2017sybilscar} implant these connections following a presumed structural property such as homophily. However, in reality, these connections are a function of two variables: (i) the attacker's strategy to send connection requests and (ii) whether a request that has been sent is accepted by the corresponding benign, which is determined by their resistance. Following this natural observation, we introduce a novel generic modeling framework and three potential attack strategies. This creates a powerful test bed for sybil detection mechanisms.

\textbf{Potential Attack Edges.} As mentioned, most sybil detection algorithms function well when the homophily property holds. Still, as observed in~\cite{gong2014sybilbelief}, this is often not the case, especially when the attacker has a powerful strategy and users are non-resistant. Thus, it would be beneficial to determine potential attack edges (sybil-to-benign connections). 
Then, they can be removed or receive less weights to enhance the homophily and consequently the performance of sybil detection algorithms.

Incoming edges for non-resistant benigns are \textit{potential} attack edges. The issue is that, in reality, we don't know which nodes are resistant/non-resistant. However, we might have an estimate of whether a node is resistant (e.g., based on their number of connections, interaction behaviors, and so on). Furthermore, we can determine whether a user is resistant by sending a sequence of requests from a dummy sybil account. However, we don't want to bombard the whole network with such dummy sybil requests. Thus, we suppose we have a fixed budget $k$ of the number of users whose resistance can be revealed. We formulate this as an optimization problem and provide an optimal linear time algorithm. 

\textbf{Discovering Benigns.} If we know that a node $v$ is benign and resistant, we can conclude that a node $u$ with an edge \textit{to} $v$ is also benign. Now, if $u$ is known to be resistant, we can conclude that a node $w$, which has an edge to $u$, is benign too, and so on. How many new benigns can be discovered if we are allowed to reveal the resistance of $k$ nodes? We prove that this problem is computationally ``hard''. However, we propose a greedy-based and traversing algorithm, which turns out to be very performant based on experiments on real-world graph data. In addition to potential attack edge discovery from above, expanding the benign set can serve as an important preprocessing step for sybil detection algorithms.

\textbf{Enhancing Detection Algorithms by Preprocessing.} Based on the contributions mentioned above, we investigate the state-of-the-art detection methods, including SybilSCAR~\cite{wang2017sybilscar}, SybilWalk~\cite{jia2017random}, and SybilMetric~\cite{asghari2022using}, with and without preprocessing steps conducted by our resistance-based mechanisms. We first examine the performance of these methods on several real-world social networks incorporated into our synthetic framework under three different attack strategies. We then gauge the performance when our potential attack edge and benigns discovery mechanisms are applied as a preprocessing step. They prove to be very impactful in enhancing the accuracy performance.

\textbf{Outline.} In the rest of this section, we provide some basic definitions and an overview of some prior work. Our data generation framework, accompanied by attack models, is provided in Section~\ref{sec:attack-model}. The maximizing benigns and potential attack edges discovery mechanisms using resistance information are given in Section~\ref{sec:discover-benign} and~\ref{sec:pae}. Finally, the experimental results of our proposed mechanisms are discussed in Section~\ref{sec:exp}.

\subsection{Preliminaries}

\textbf{Graph Definition.} A social network is represented by a directed graph $G=(V,E)$ with node set $V$ and edge set $E \subseteq V \times V$. We have $|V|=n$ and $|E|=m$. Nodes correspond to users; edges represent connections, such as following (directed) or friendship (bidirectional).

We define $\Gamma_{\text{in}}(v) := \{u : (u,v) \in E\}$
and
$\Gamma_{\text{out}}(v) := \{u : (v,u) \in E\}$ to be respectively the set of incoming and outgoing neighbors of a node $v$. Then, we define $d_{\text{in}}(v) := |\Gamma_{\text{in}}(v)|$,
$d_{\text{out}}(v) := |\Gamma_{\text{out}}(v)|$
and
$d(v) := d_{\text{in}}(v) + d_{\text{out}}(v)$. Let $\Delta_{\text{in}} = \max \left\{ d_{\text{in}}(v) \mid v \in V \right\}$ denote the maximum incoming degree and $\Delta_{\text{out}} = \max \left\{ d_{\text{out}}(v) \mid v \in V \right\}$. 
We show an edge from $u$ to $v$ by $e_{uv}$, and if the edge is undirected, we show that by $\bar{e}_{uv}$. In addition, for two subsets $V_1,V_2 \subseteq V$ we define $\partial(V_1, V_2) := \{e_{vu} \in E: v \in V_1 \wedge  u \in V_2 \}$.
A path in a graph $G = (V, E)$ is a sequence of nodes $v_1, v_2, \ldots, v_k$ where each adjacent pair is connected by an edge, i.e., $(v_i, v_{i+1}) \in E$ for $1 \leq i \leq k-1$.

\textbf{Benign and Sybil Classification Problem.} In this paper, we aim to provide preprocessing algorithms which help mechanisms solve the following problem. 
\tcbset{colback=lightblue, colframe=blue, 
        width=\columnwidth, boxrule=0.5pt, arc=3mm, 
        left=1mm, right=1mm, top=1mm, bottom=1mm}

\textsc{\textbf{Benign and Sybil Classification (BSC) Problem}}\\
\textbf{Input:} Given a graph $ G = (V, E) $ where nodes are labeled as \textit{Benign}, \textit{Sybil}, or \textit{Unknown}, partitioned into subsets
 $ B $, $ S $, and $ U $ respectively ($ B \cup S \cup U = V $ and $(B \cap S) \cup (B \cap U) \cup (S \cap U) = \emptyset$). 
 
\noindent
\textbf{Goal:} Maximize the number of correctly labeled nodes in $U$.

\textbf{User Resistance.} In the real world, each user can choose to accept or reject sybil requests. Some users are more cautious and reject such requests, while others might accept them. For simplicity, we assume each user, which is represented as a node in a graph, either accepts all sybil requests or rejects all of them. To model this, we define a binary value \( r: v \rightarrow \{0,1\} \) to represent a user's resistance to sybil requests, where $1$ means resistant (rejecting) and $0$ means non-resistant (accepting). 
Since, in reality, the value of $r(v)$ is not known to us, 
we suppose we know a probability function \( p_r(v) \) whose value is more likely to be closer to $1$ if $v$ is resistant and $0$ otherwise.
Such estimates can be achieved in practice by studying certain user features, such as the number of incoming/outgoing connections, the number of known sybil/benign neighbors, and content posted. In our paper, we suppose such a probability distribution is given to us as input.

\subsection{Related Work}
Some mechanisms aim to detect sybils at the early stages. This could be as early as registration time using techniques like captchas or analyzing data such as registration time and IP addresses, cf.~\cite{yuan2019detecting,liang2021unveiling}. The sybils that pass through these filters might be caught based on their first activities, such as friendship requests, cf.~\cite{breuer2023preemptive}. 

While such early detection mechanisms are useful in identifying a considerable fraction of sybils, more sophisticated attacks can deceive them by mimicking benign users. Thus, there has been a growing interest in leveraging network structure information to detect more sophisticated attacks, cf.~\cite{yoon2021graph,jia2017random}. This is because, unlike information such as username or IP, sybil users have very limited power to control their network structural properties. Below, we give an overview of various existing methods.

\textbf{Graph Metrics.} One natural approach is finding graph metrics that separate sybils and benigns. Asghari et al.~\cite{asghari2022using} analyzed the relationship between being sybil or benign based on metrics such as degree, betweenness, eigenvector centrality, clustering coefficient, and average shortest path length. Yoon~\cite{yoon2021graph} defined the graph accessibility metric for nodes to identify sybils. Jethava and Rao~\cite{jethava2022user} proposed a hybrid approach, combining user behavior analysis with graph-based techniques. They used metrics such as the Jaccard index and betweenness centrality.

\textbf{Random Walk.}
Random walk-based detection mechanisms have particularly become popular, cf.~\cite{boshmaf2015integro,zhang2022enhancing}. They usually rely on the premise that a random walk starting from a sybil user is more likely to encounter sybil users. SybilGuard~\cite{yu2006sybilguard} was one of the first methods to leverage random walks for sybil detection. SybilLimit~\cite{yu2008sybillimit} improved on this by enhancing scalability and providing a tighter bound on the number of accepted sybils.
SybilInfer~\cite{danezis2009sybilinfer} additionally used Bayesian inference and Monte Carlo simulations to provide a more robust mechanism.
SybilRank~\cite{cao2012aiding} proposed using a ranking algorithm that prioritizes nodes based on the degree-normalized probability of a short random walk from a non-sybil landing on them.
SybilWalk~\cite{jia2017random} improved random walk-based models by incorporating known benigns and sybils simultaneously. This model defines two extra nodes as label nodes and then connects known nodes to their respective label node.

\textbf{Belief Propagation.} Another popular approach, cf.~\cite{pearl1988probabilistic}, is to assign some initial probability of being sybil to each node (where known sybils get higher and known benigns get lower probabilities) and then use some updating mechanism, called belief propagation, to improve these probabilities. The idea is that a node's probability (belief) is updated as a function of the probabilities of its neighbors.
SybilBelief~\cite{gong2014sybilbelief} is a semi-supervised learning algorithm that uses belief propagation to detect sybils.
SybilSCAR~\cite{wang2017sybilscar} unified random walk and belief propagation methods, applying local rules iteratively to identify sybils.
GANG~\cite{wang2017gang} improved previous models by using directed graphs.
SybilFuse~\cite{gao2018sybilfuse} employed collective classification by training local classifiers to calculate trust scores first, then propagating these scores using weighted random walk and belief propagation to improve detection accuracy.

\textbf{Machine Learning.} Machine Learning (ML) methods for sybil detection usually involve feature extraction, model training, and classification. Some basic ML methods include SVM, Logistic Regression, and Random Forest, cf.\cite{van2018using, kondeti2021fake, le2020application}.
Furthermore, deep learning methods can process raw data and capture complex structures, making them ideal for sybil detection. Goyal et al.~\cite{goyal2023detection} used graph convolutional network (GCN) to learn structural features, while Borkar et al.~\cite{borkar2022real} employed recurrent networks for text content processing, followed by clustering for sybil detection. Recent advances leverage graph neural networks (GNNs) and attention mechanisms. Yang and Zheng~\cite{yang2020fake} used attention-based GNNs, Khan et al.\cite{khan2024graph} developed a GNN-based framework to analyze user profiles and connections, and Liu et al.\cite{liu2018heterogeneous} integrated diverse attributes using heterogeneous GNNs.

\section{Attack Models}
\label{sec:attack-model}

To effectively evaluate detection algorithms, it is essential to provide high-quality datasets. Many existing datasets suffer from a lack of proper labeling, as they are often labeled manually and are quite limited in size. To tackle this issue, some synthesized datasets have been provided~\cite{gong2014sybilbelief,lu2023sybilhp}
but they still fall short in various aspects. For example, \cite{gong2014sybilbelief} uses a synthesized model to generate the sybil area, which is less favorable than using real-world data. Furthermore, the existing models~\cite{lu2023sybilhp} make special assumptions on the strategy employed by the attacker and the structure of the edges between sybil and benign parts, for example, uniform random edges. They then leverage these structural properties to devise algorithms that are tailored for this setting. However, the proposed algorithms should ideally be agnostic of the attack strategy deployed by the attacker since the attacker can adopt various attack strategies. Moreover, the availability of diverse datasets generated by various attack strategies enhances the reliability of evaluation outcomes. This encourages us to synthesize datasets with different characteristics. We will see that the notion of user resistance, defined in this paper, forms an excellent foundation for synthesizing realistic datasets.

To synthesize a labeled network, one needs to define the nodes and their labels and the edge set, especially \textit{attack edges}, formed from sybils to benigns. An existing dataset of online social networks is usually used as the benign set, treating all individuals as benign~\cite{wang2017sybilscar}. For the sybil set, a common approach is that the attacker may replicate a subgraph from the real graph to secure an initial real-world structure. If attackers avoid using this strategy, their presence becomes readily detectable through the analysis of simple graph features.

Following the above description, the sybil set $S$ will be created based on the chosen subset $B' \subset B$, where $B$ is the benign set. Then, the edges between nodes in $S$ are formed to copy the corresponding structure in $B'$.
Each node $v \in B'$ has a corresponding copy in $S$ and vice versa; these nodes are referred to as \textit{dual} of each other.

Following a specific strategy, the attacker sends requests to the benigns. If a benign has resistance $1$, the request is rejected; otherwise, it is accepted.
Recall that we show the resistance of each node $v$ with $r(v)$, which is set to $0$ (non-resistant) or $1$ (resistant). For example, if the resistance of all nodes is set to 1, all the sent requests are rejected, and no edge is formed from $S$ to $B$. However, in reality, not all nodes are resistant.

Let's note that while our algorithms start from known benign nodes and do not consider sybil resistance, sybils can indeed have resistance in real-world scenarios. Attackers might program sybils to reject specific requests to mimic realistic behavior or avoid detection methods that use test requests. Additionally, sybils might reject requests from nodes suspected to be other sybils, especially without collaboration among attackers, to avoid higher homophily and a higher risk of being detected.

As stated, our proposed framework is not limited to a particular attack strategy and facilitates the study of various strategies. Below, we provide three natural approaches which may be used by the attacker. Before that, we explain some concepts which are shared among them.

Assume the number of \textbf{A}ttack \textbf{E}dges from node $s_i\in S$ to $B$ is $AE(i)$. As discussed above, we suppose that the attacker ``copies'' the subgraph among a subset of benigns $B'\subset B$ as $S$. To mimic a benign behavior further, one natural strategy is to ensure that each node in $S$ has as many edges to $B\setminus B'$ as its dual. More precisely, one aims to have $AE(i) =|\partial(\{dual(s_i)\}, B\setminus B')|$.

Furthermore, it is plausible to consider adding edges from the benign region ($B$) to the sybil region ($S$). Given the low likelihood that real users spontaneously connect to sybils, we posit that only nodes that have already accepted an edge from a sybil user might reciprocally connect back to it. Formally, consider each edge $(u, v) \in E$ where $u \in S$ and $v \in B$. We add the reverse edge $(v,u)$ with a certain probability, say $1/2$, for each such edge.

Below, we describe our proposed attack strategies. Please refer to Appendix~\ref{appendix:Algorithms of Attack Strategies} for a detailed pseudocode.

\textbf{Random Attack Strategy.} In this strategy, the attacker simply sends random requests to benigns. 
The number of requests from each node $i\in S$ is $c\cdot AE(i)$. For an appropriate choice of constant $c$, as a function of the average resistance in $B$, the expected number of attack edges from $i$ will be our desired value $AE(i) =|\partial(\{dual(s_i)\}, B\setminus B')|$ (the same applied to the strategies below).

\textbf{Preferential Attachment Attack Strategy.} The attacker might prefer to target nodes with the least resistance, especially those that have accepted more attack requests in the past. This also aligns with the \textit{power-law} behavior observed in real-world social networks. Thus, we leverage the preferential attachment model of network growth.
We use a modified version of the preferential attachment model based on the Barabási-Albert (BA) graph~\cite{barabasi1999emergence}, named \textit{Modified BA}. This takes into account both the original degree of each node and the number of accepted attack requests. This modification has been applied to take into account that the attacker is more likely to send requests to nodes that have already proven vulnerable (by accepting previous sybil requests).

\textbf{BFS Attack Strategy.} To enhance the difficulty of detecting sybil behavior, another natural approach is that each sybil attempts connections with the neighbors of its dual in the benign region. The acceptance of these connection requests, however, depends on the targeted node's resistance. To achieve the desired number of attack edges for each sybil node $s_i$ (that is, the aforementioned value of $AE(i)$), we process nodes using a BFS algorithm. For each node $s_i$, the BFS starts from $dual(s_i)$. 
If the number of nodes that BFS can reach is not enough, a preferential attachment attack strategy is utilized to complete the process.

\section{Preprocessing Algorithms}
\label{section:preprocessing}
In this part of the paper, we will study the problems of Maximizing Benigns
(in Section~\ref{sec:discover-benign}) and Discovering Potential Attack Edges (in Section~\ref{sec:pae}) leveraging the resistance probability information. We propose algorithms for solving these problems. As will be discussed in Section~\ref{sec:exp}, the outcome of these algorithms serves as a very valuable preprocessing step for detection algorithms.

\subsection{Maximizing Benigns}
\label{sec:discover-benign}

\textsc{\textbf{Maximizing Benigns (MB) Problem}}\\
\textbf{Input:} Given $G = (V, E)$, sets of benigns and sybils \(B\) and \(S\) such that $B \cap S = \emptyset$, budget $k$, and resistant probability distribution $p_r: V\rightarrow [0,1]$. 

\noindent
\textbf{Goal:} Maximize the number of newly discovered benigns by revealing the resistance of $k$ nodes.

An algorithm is permitted to reveal the resistance of nodes in a node set $A$ of size $k$ as a reveal set. Each selected node $u$ will be revealed to be resistant ($r(u)=1$) or not, with probability $p_r(u)$, independently. (In real life, we send a set of friend requests from some dummy sybils to a node $u$ and determine whether it is resistant, but our budget has a bound to avoid bombarding all users with such requests.) Then, the goal is to maximize the number of newly discovered benigns.

\begin{observation}
\label{obs:pathtoB_toDescovered}
A node $u$ in $V\setminus (B\cup S)$ can be newly discovered benign if and only if it has a path to a node in $B$ and each node $v$ on that path (except, potentially, $u$ itself) has revealed to have $r(v)=1$.
\end{observation}

\begin{definition} 
$f(A)$ is the expected number of discovered benigns upon revealing the nodes in the reveal set $A$.
\end{definition}

\subsubsection{Hardness Results}

We prove that our problem is computationally hard by a reduction from the Maximum Coverage Problem.

\begin{definition} [Maximum Coverage (MC) Problem]
Given set $W=\{w_1,w_2,\cdots,w_h\}$, collection of subsets $Q=\{q_1,q_2,\cdots,q_l\}$,  $\forall q_i \in Q: q_i \subseteq W$, and budget $k$, what is the maximum number of elements which can be covered by a set $A\subset Q$ of size $k$? We say a set $q_i$ covers an element $w_j$ if $w_j\in q_i$.
\end{definition}

\begin{theorem}[\cite{Fiege1998Athreshhold}]\label{thm:apx_mc}
There is no \((1-\frac{1}{e})\)-approximation polynomial time algorithm for the MC problem unless \(\text{NP}\subseteq \text{DTIME}(n^{O(\log \log n)})\).
\end{theorem}

Assume $I = \langle W, Q, k \rangle$ is an instance of MC. We define a transformation to convert \( I \) into $I'=\langle G,B,S,k,p_r(\cdot) \rangle$, which is an instance of MB.
Without loss of generality, we assume that each \( w_i \in W \) appears in at least one of the subsets \( q_j \in Q \) (otherwise, they could be simply ignored because they cannot be covered at all). Now, we present the transformer in the following. 

\textbf{Transformer.}
To define the transformer to convert \( I \) to \( I' \), we first need to construct the MC problem in graph space as $G=(V,E)$. We define $V=V_Q\cup V_W$, where $V_Q=\{v_{q_i}:  1\le i\le l\}$ and $V_W=\{v_{w_i}: 1\le i\le h\}$. Note that $V_Q \cap V_W = \emptyset$. We also define $E= \{(v_{w_i},v_{q_j}):w_i \in q_j$\}. In addition, $k$ is the same,
$S = \emptyset$, $B=V_Q$, and $p_r(v) = 1 $ for all nodes (which means we know all nodes are resistant). An example is provided in Appendix~\ref{appendix:Connection Between Optimal Solutions example}.

\textbf{Connection Between Optimal Solutions.}
For any arbitrary instance of the MC problem, let $OPT_{\text{MC}}$ denote the optimal solution, and similarly $OPT_{\text{MB}}$ for the MB problem. We claim that
\begin{equation}
    \label{eq:MC=MB}
    OPT_{\text{MC}} = OPT_{\text{MB}}
\end{equation}

First, let $A$ be an optimal solution for the MB problem. Note that $A\cap V_W=\emptyset$ because revealing a node in $V_W$ does not result in discovering any new benign. Consider that set $A_q$, which includes a set $q_i$ if and only if $v_{q_i}\in A$. If a node $v_{w_j}$ is newly discovered benign by revealing $A$, then there is a node $v_{q_i}$ such that $(v_{w_j},v_{q_i})\in E$ and $v_{q_i}\in A$. This implies that $w_j\in q_i$ and $q_i\in A_q$. Therefore, $w_j$ is covered by $A_q$. This implies that $OPT_{\text{MC}} \ge OPT_{\text{MB}}$.

It remains to prove that $OPT_{\text{MC}} \le OPT_{\text{MB}}$. Let $A_q\subset Q$ generate an optimal solution for the MC problem. Define $A$ to include a node $v_{q_i}$ if and only if $q_i\in A_q$. Consider an element $w_j$, which is covered by $A_q$, and let $q_i$ be a set which covers it. Then, $v_{q_i}\in A$ and $(v_{w_j},v_{q_i})\in E$. Thus, $w_j$ will be a newly discovered benign since $w_j\in V\setminus B$, and it has an edge to a node with resistance $1$. This implies that $OPT_{\text{MC}} \le OPT_{\text{MB}}$.

\textbf{Inapproximability.}
Assume there exists a polynomial-time \((1 - \frac{1}{e} )\)-approximation algorithm \(Alg_{\text{MB}}\), that solves the MB problem. For an instance of the MC problem, we can use the previously mentioned transformer to construct an instance of our problem in polynomial time. Then, we apply \(Alg_{\text{MB}}\) to solve the constructed instance, yielding a solution \(Sol_{\text{MB}}\). Using the same argument as previously described for the relationship between optimal solutions, we can show that 
$Sol_{\text{MC}} \geq Sol_{\text{MB}}$. Thus,
$Sol_{\text{MC}} \geq Sol_{\text{MB}} \geq (1-\frac{1}{e})OPT_{\text{MB}}$. With the help of Equation~\eqref{eq:MC=MB} we can conclude
$Sol_{\text{MC}} \geq (1-\frac{1}{e}) OPT_{\text{MC}}$

Therefore, we get an \((1 - \frac{1}{e} )\)-approximation algorithm for MC, which operates in polynomial time. However, according to Theorem~\ref{thm:apx_mc}, this is impossible unless 
\(\text{NP}\subseteq \text{DTIME}(n^{O(\log \log n)})\). Thus, we have the theorem below.

\begin{theorem}
\label{thm:hardness}
There is no \((1-\frac{1}{e})\)-approximation polynomial time algorithm for the MB problem unless \(\text{NP}\subseteq \text{DTIME}(n^{O(\log \log n)})\).
\end{theorem}

\subsubsection{Greedy Approach}

For the MB problem, we can use the classical greedy algorithm, which iteratively reveals the node that maximizes the expected number of newly discovered benigns (an exact description is given in Appendix~\ref{appendix:mbgreedy}). According to Nemhauser et al.~\cite{nemhauser1978analysis}, if $f(\cdot)$ is non-negative, monotone and sub-modular, then this algorithm is $(1 - \frac{1}{e})$-approximation (which matches the bound from Theorem~\ref{thm:hardness}).

However, we prove that the objective function $f(\cdot)$, in fact, is \textit{not} submodular (please see Appendix~\ref{Appendix:Proof for Nonsubmodularity} for a proof). Thus, we cannot conclude $(1 - \frac{1}{e})$-approximation guarantee. On the positive side, the greedy algorithm has proven to be performant, cf.~\cite{bian2017guarantees,guo2019targeted}, even if the submodularity property doesn't hold. Thus, we study this further. However, it turns out that computing $f(\cdot)$ is \#P-hard, implying that the greedy algorithm is not polynomial time. 
To tackle this issue, we show that we can estimate $f(\cdot)$ with an arbitrarily small error parameter using the Monte Carlo method.

\paragraph{\#P-hardness.} We rely on a reduction from $a$-$b$ Connectedness for Induced Subgraphs Problem, which is known to be \#P-hard~\cite{valiant1979complexity}.

\noindent\textit{$a$-$b$ Connectedness for Induced Subgraphs (CIS):}
\\\textit{Input:} Given a directed graph $G=(V,E)$ and $a,b \in V$.
\\\textit{Goal:} What is the number of induced subgraphs of $G$ such that there is a path from $a$ to $b$.

\begin{theorem}\label{thm:sharpphard}
Computing $f(\cdot)$ is \#P-hard.
\end{theorem}

\begin{proof}

Suppose $G=(V, E)$ and $ a, b\in V$ are given as the input of CIS problem. We consider two instances of our problem to compute $f(\cdot)$. First consider graph $G=(V,E)$,
$B = \{b\}$, 
$S = \emptyset$,
$\left\{\begin{matrix}
r(v) = \frac{1}{2} & v \in V \setminus \{a,b\}\\ 
r(v) = 1 & v \in \{a,b\}
\end{matrix}\right.$
, and 
reveal set $A = V$. The second case is identical to the first case, except we add a node $a'$ and add an edge from $a'$ to $a$. Let's call this graph $G'=(V',E')$, where $V'=V\cup \{a'\}$. Furthermore, again $B=\{b\}$, $S=\emptyset$, the values of $r(v)$ are the same for all nodes, and $r(a')$ is any arbitrary value (its choice doesn't impact our argument), and $A=\{V\}$. See Figure~\ref{fig:sharpphard} for a visualization.

\begin{figure}[htbp]
    \centering
    \includegraphics[width=0.95\linewidth]{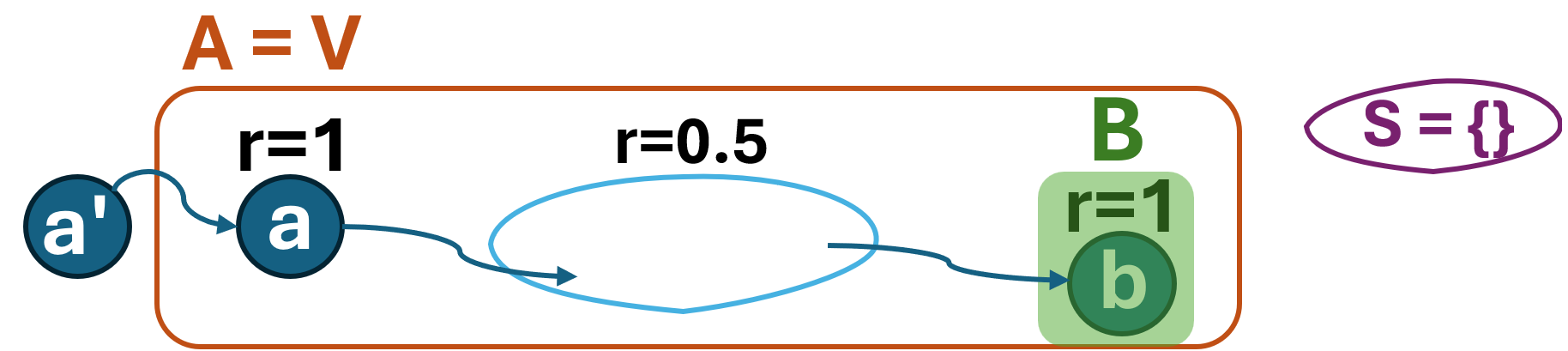}
    \caption{The construction used in the proof of \#P-hardness.}
    \Description{The configuration of nodes and their resistance is provided. Three main nodes are $a'$, $a$, and $b$. This figure is for the proof of \#P-hardness.}
    \label{fig:sharpphard}
\end{figure}

Let $f_{-a'}(A)$ and $f_{+a'}(A)$ denote the value of $f(A)$, the expected number of discovered benigns, in the first and second case, respectively. Firstly, we observe that $a'$ is discovered to be benign if and only if $a$ is discovered to be benign. This implies that $f_{+a'}(A)-f_{-a'}(A)$ is equal to the probability of the node $a$ being discovered to be benign. Secondly, note that each of the nodes in $V\setminus \{a,b\}$ is revealed to be resistant or not with probability $1/2$, independently. Thus, based on Observation~\ref{obs:pathtoB_toDescovered}, the probability that $a$ is revealed to be benign is the same as having a path from $a$ to $b$ once we remove each node in $V\setminus \{a,b\}$ with probability $1/2$. Combining the aforementioned two points, we can conclude that $f_{+a'}(A)-f_{-a'}(A)$ is equal to the fraction of all $2^{n-2}$ possible induced subgraphs obtained from removing nodes in $V\setminus \{a,b\}$ in graph $G$, where there is a path from $a$ to $b$.

So far we concluded that the solution to CIS problem is equal to $(f_{+a'}(A)-f_{-a'}(A))\times 2^{n-2}$. Since the provided transformer is polynomial time and CIS is \#P-hard, we conclude that the problem of computing $f(\cdot)$ is \#P-hard as well.  
\end{proof}

\paragraph{Monte Carlo Greedy Algorithm}

Algorithm~\ref{algorithm:montecarlo} modifies the classical greedy approach by using the Monte Carlo method to estimate $f(\cdot)$ since, as we proved, its computation is \#P-hard. The function \texttt{EST($\cdot$)} runs the revealing process $R$ times and computes the number of discovered benigns each time. Then, the average values found are returned as an estimation of the original expectation.

As part of the function \texttt{EST($\cdot$)}, to compute the number of discovered benigns based on $r(\cdot)$ and the reveal set $A$, we apply a BFS starting from known benigns in $A$. During the BFS, only nodes with resistance $1$ (starting from $A$) are added to the queue, marking their neighbors as discovered. This ensures every discovered node has a path of nodes with resistance $1$ to an initial benign, meeting the criterion in Observation~\ref{obs:pathtoB_toDescovered}. The pseudocode of this algorithm is provided in Appendix~\ref{appendix:BFS-based Discovery Algorithm}.

\begin{algorithm}[htbp]
\caption{Monte Carlo Greedy Algorithm}
\label{algorithm:montecarlo}
\raggedright  \textbf{Input} $G=(V,E),B, S, \text{budget }k,p_r(\cdot),\epsilon, \alpha$
\\
\raggedright  \textbf{Output} $reveal\_set\ A$
\\

\begin{algorithmic}[1]

\State Initialize $A \gets \emptyset$

\For{$i = 1$ to $k$}
    \State $v \gets \arg\max_{w \in V\setminus A} \Call{Est}{A \cup \{w\},p_r(\cdot),B,S, \epsilon,\alpha}$
    \State $A \gets A \cup \{v\}$
\EndFor
\State \Return $A$

\item[]
\Function{Est}{$A, p_r(\cdot), B, S, \epsilon, \alpha$}
    \State Initialize $count \gets 0$
    \State $R \gets \left\lceil  \frac{k^2\Delta_{\text{in}}^2 \ln(\frac{1}{1-\alpha})}{2\epsilon^2} \right\rceil$
    \For{$R$ iterations} 
        \State Assign resistance $r(v)$ for each $v \in A$ based on $p_r(v)$
        \State $new\_benign \gets$ number of discovered benigns based on $B$, $S$, and assigned resistances.   
        \State $count \gets count + new\_benign$ 
    \EndFor
    \State \Return $count / R$
    
\EndFunction

\end{algorithmic}
\end{algorithm}

Furthermore, using Hoeffding's inequality~\cite{hoeffding1994probability}, one can achieve arbitrary error margin $\epsilon$ and gain confidence $\alpha$ by setting the number of iterations $R \ge \frac{k^2\Delta_{\text{in}}^2 \ln(\frac{1}{1-\alpha})}{2\epsilon^2}$. 
The time complexity of the algorithm is $O(k \cdot R \cdot (n^2 + nm))$. 
One can check that, for constant $\epsilon$ and $\alpha$, the run time is polynomial.
Please refer to Appendix~\ref{appendix:hoeffding Proof of bound} and~\ref{appendix:complexity discussion} for more details on these two points.

\subsubsection{Proposed Traversing Algorithm}
\label{subsection:Proposed Traversing Algorithm}

We propose an algorithm that iteratively selects the node expected to reveal the maximum number of newly discovered benigns. More precisely, we keep a set of nodes $N$ such that if a node in $N$ is revealed to be resistant, then we can conclude all its incoming neighbors are benign. Based on Observation~\ref{obs:pathtoB_toDescovered}, these nodes already have a path to a node in $B$, and all nodes on the path are revealed to be benign. Thus, initially, $N$ is simply $B$, but it is updated as the algorithm reveals more resistant nodes. However, among all the incoming neighbors of a node $v$, we are only interested in the ones that will be ``newly'' discovered benign. Thus, we define $\hat{\Gamma}(v)$, which excludes unrelated nodes such as the already discovered ones (that is, $B\cup S$) or already discovered benigns. We continuously update $\hat{\Gamma}(v)$ as more nodes are revealed and added to $A$. More precisely, for each node $v$, if $r(v)$ is revealed to be 1, then we must consider every in-neighbor $u$ of $v$, and for each out-neighbor $w$ of $u$, we then remove $u$ from the set of in-neighbors of $w$. This is because $u$ has already been discovered to be benign. The pseudocode of this algorithm is provided in Algorithm~\ref{algorithm:TraversingResistanceDegreeAware}.

The \texttt{For} loop runs in $O(n+m)$. The \texttt{While} loop is computed $k$ times, and the lines 7-9 can clearly be executed in $O(n)$. The complexity of the \texttt{If} statement part is $O(\Delta_{\text{in}} \cdot \Delta_{\text{out}})$. This is achieved by using the appropriate data structure (please refer to Appendix~\ref{appendix:TraversingResistanceDegreeAware} for more details). Thus, the time complexity of this algorithm is $O((n+m)+ k\cdot (n+\Delta_{\text{in}} \cdot \Delta_{\text{out}}))$, which can be bounded by $O(kn^2)$. 

\begin{algorithm}[htpb]
\caption{Traversing Algorithm}
\label{algorithm:TraversingResistanceDegreeAware}
\raggedright  \textbf{Input} $G=(V,E), B, S,\text{budget }k,p_r(\cdot)$
\\
\raggedright  \textbf{Output} $reveal\_set\ A$
\begin{algorithmic}[1]
\State $A \gets \emptyset$
\State $ N \gets B$
\For{$v \in V$}
\State $\hat{\Gamma}_{\text{in}}(v) \gets \Gamma_{\text{in}}(v) \setminus (B \cup S)$
\State $\gamma_{\text{in}}(v) = |\hat{\Gamma}_{\text{in}}(v)|$
\EndFor

\While{$|A|<k$}
    \State pick $v$ with highest $p_r(v) \cdot \gamma_{\text{in}}(v)$ between nodes in $N$
    \State $A \gets A \cup \{v\}$
    \State $N \gets N \setminus \{v\}$
    \If{$r(v)=1$}
        \State $N \gets N \cup \hat{\Gamma}_{\text{in}}(v)$
        \For{$u \in \hat{\Gamma}_{\text{in}}(v)$}
            \For{$w \in \Gamma_{\text{out}}(u)$}
                \State $\hat{\Gamma}_{\text{in}}(w) \gets \hat{\Gamma}_{\text{in}}(w) \setminus \{u\}$
                \State $\gamma_{\text{in}}(w) \gets \gamma_{\text{in}}(w) - 1$
            \EndFor
        \EndFor
    \EndIf
\EndWhile
\State \Return $A$

\end{algorithmic}
\end{algorithm}

\subsection{Discovering Potential Attack Edges}
\label{sec:pae}
\begin{definition} An edge $e=(u,v)$ is a \textbf{P}otential \textbf{A}ttack \textbf{E}dge (PAE) if $u \in V \setminus (B \cup S)$, $v \in B$, and $r(v)=0$. In other words, it is an incoming edge to a benign $v$ with $r(v)=0$.
\end{definition}


\textsc{\textbf{Discovering Potential Attack Edges Problem}}\\
\textbf{Input:}
Given $G = (V, E)$, sets of benigns and sybils \(B\) and \(S\) such that $B, S \subseteq V$, $B \cap S = \emptyset$, budget $k$, and probability of resistant $p_r: V\rightarrow [0,1]$.

\noindent
\textbf{Goal:} Maximize expected number of discovered potential attack edges by revealing $r(v)$ of $k$ nodes from $B$.

For each node $v$, we know $d_{\text{in}}(v)$ and $p_r(v)$, and $(1-p_r(v)) \cdot |\Gamma_{\text{in}}(v) \setminus (B\cap S)|$ is the expected number of edges which will be discovered by revealing the node $v$. Since there is no overlap between edges discovered by each revealed node, choosing $k$ nodes with the highest value of $(1-p_r(v)) \cdot |\Gamma_{\text{in}}(v) \setminus (B\cap S)|$ is the optimal solution. This is outlined in Algorithm~\ref{pseudocode:RPAE}. For choosing the top $k$ nodes based on the computed values, we can use the median of medians algorithm~\cite{BLUM1973Time}, which runs in $\mathcal{O}(n)$. Thus, the overall time complexity is linear.

\begin{algorithm}[htbp]
\caption{Proposed Algorithm: Select Top $k$ Nodes}
\label{pseudocode:RPAE}
\raggedright  
\textbf{Input} \textbf{Input} $G=(V,E), B, S,\text{budget }k,p_r(\cdot)$
\\
\raggedright  \textbf{Output} List of $k$ nodes

\begin{algorithmic}[1]

\State $node\_values \gets \emptyset$
\For{each node $v$ in $B$}
    \State  $\textit{value} \gets (1-p_r(v)) \cdot |\Gamma_{\text{in}}(v) \setminus (B\cap S)|$ 
    \State $node\_values \gets node\_values \cup \{(v, \textit{value})\}$ 
\EndFor
\State $top\_k\_nodes \gets \text{top $k$ nodes of } node\_values \text{ based on } value$ using median of medians algorithm
\State \Return \textit{top\_k\_nodes}
\end{algorithmic}
\end{algorithm}

As we will discuss, we can reduce the weight of discovered PAEs as a preprocessing step to assist sybil detection algorithms.

\section{Experiments}
\label{sec:exp}

After establishing our experiments setting in Sections~\ref{sec:exp-setup} and~\ref{sec:attack-strategies}, we test our proposed algorithms for maximizing benigns problem and potential attack edges problem in Section~\ref{sec:preprocessing}. Then, in Section~\ref{sec:classification}, we study the impact of applying our algorithms as a preprocessing step for various sybil/benign classification methods.

\subsection{Experimental Setup}
\label{sec:exp-setup}

We use real-world graph data as the benign set. Then, our attack strategies
are used to generate sybils and their connections to sybil/benigns. For the graph, the data used include Facebook, LastFM, and Twitter from SNAP~\cite{snapnets2014}.
In addition, we use, Pokec network
\footnote{\url{https://github.com/binghuiwang/sybildetection/blob/master/Directed_Pokec.rar} (accessed July 20, 2024)} 
(we use only the subgraph induced by benigns). Some statistics of these networks are presented in 
Appendix~\ref{appendix:Datasets Real world Statistics reports}. For undirected networks, we assume the edges in both directions exist.

After generating sybils and attack edges in each dataset, as explained in the next section, we randomly selected $2\%$ of benigns and an equal number of sybils for the training set. The remaining sybils were included in the test set, along with an equal number of benigns. 
This means the size of known benigns for Facebook, Pokec, LastFM, and Twitter datasets are $80$, $200$, $150$, and $200$, respectively.

\subsection{Attack Strategies}
\label{sec:attack-strategies}

To choose the parameters in our attack models, we rely on some real-world statistics. Facebook reports indicate that the proportion of sybils is around $16\%$~\cite{moore2023fake}.
A 2013 study also found that $10\%$ of Twitter users were fake~\cite{wagstaff2013twitter}. Thus, we choose the size of the sybil set to be around $10\%$ of the network in our attack strategies.

Furthermore, Vishwanath~\cite{vishwanath2018fake} conducted a study involving the creation of fake Facebook profiles and sending friend requests to students in a large university. It was observed that only $30\%$ of students declined friend requests from fake Facebook profiles, $52\%$ were undecided after two weeks, and $18\%$ accepted immediately. Based on these results, we chose $25\%$ of nodes to be non-resistant.

To compute the probability of resistance $p_r(v)$ for a node $v$ based on resistance $r(v)$, we generate a random number $x$ between $0$ and $1$ and then we use \(p_r(v) =  (1-r(v))x^3+r(v)(1-x^3) \). We choose this simple formula to introduce randomness while giving higher probabilities when \( r(v) = 1 \) and lower probabilities when \( r(v) = 0 \). For example, for $r(v)=1$, the probability of \( p_r(v) \geq 0.5 \) is $0.79$, and the average probability value is $0.75$. In addition, in all attack strategies, we set $c = 4$ (please refer to Section~\ref{sec:attack-model} for the parameters of the attack models) since we aim for a non-resistance ratio of $25\%$. This way, the number of accepted attacks for a sybil~$s_i$ in expectation will be $\frac{25}{100} \cdot c \cdot AE(i) = AE(i)$, which is our desired value. Some statistics on the outcome of the attack strategies are provided in the Appendix~\ref{appendix:Additional Datasets after Attacks Statistics reports}.

\subsection{Preprocessing}
\label{sec:preprocessing}
In this section, we evaluate our proposed algorithms for maximizing benigns and potential attack edges problems.

\subsubsection{Maximizing Benigns}

We compare our proposed algorithms, Monte Carlo Greedy (Algorithm~\ref{algorithm:montecarlo}) and Traversing (Algorithm~\ref{algorithm:TraversingResistanceDegreeAware}), against the following baseline methods.

\textbf{1. Random.}
Randomly select $k$ nodes from benign set $B$ to reveal.

\textbf{2. Highest-Resistance.}
Pick $k$ nodes with highest $p_r(\cdot)$ from $B$.

\textbf{3. Highest-Resistance-and-Degree.}
Pick $k$ nodes $v$ with highest $p_r(v).|\Gamma_{\text{in}}(v) \setminus (B \cup S)|$ from $B$.

Between our two proposed algorithms, Traversing reveals the resistance of each selected node and uses that outcome information when selecting the next nodes, while the Monte Carlo Greedy does not. To make this a fairer comparison, we also consider a Monte Carlo Greedy variant, which reveals each selected node's resistance. This is called Resistance Aware Monte Carlo Greedy (please refer to Appendix~\ref{Appendix:montecarloresistanceaware} for an exact description of this algorithm).

Figure~\ref {fig:mb_facebook} illustrates the performance of different algorithms on the Facebook dataset. (The high variance in the results of the Monte Carlo Greedy algorithms in this experiment is attributed to the fact that the run was performed only once.) We observe that our proposed algorithms, Traversing and two variants of Monte Carlo, significantly outperform other algorithms. Among our proposed algorithms, the Traversing algorithm performs better than the Monte Carlo Greedy algorithms. Furthermore, based on our experiments, the Traversing algorithm is substantially faster than the two Monte Carlo Greedy algorithms. For example, for the Facebook dataset and Preferential Attachment attack strategy, with a budget of $k = 30$, the Traversing algorithm completes in $119$ milliseconds, while the Monte Carlo Greedy algorithm takes $52$ minutes. Thus, the Traversing algorithm is not only more \textit{accurate}, but also \textit{significantly faster}. 
For other datasets, similar results are provided in the Appendix~\ref{appendix:Additional MB experiment reports}. 

\begin{figure*}[htbp]
    \centering
    \begin{subfigure}[b]{0.32\textwidth}
        \centering
        \includegraphics[width=\linewidth]{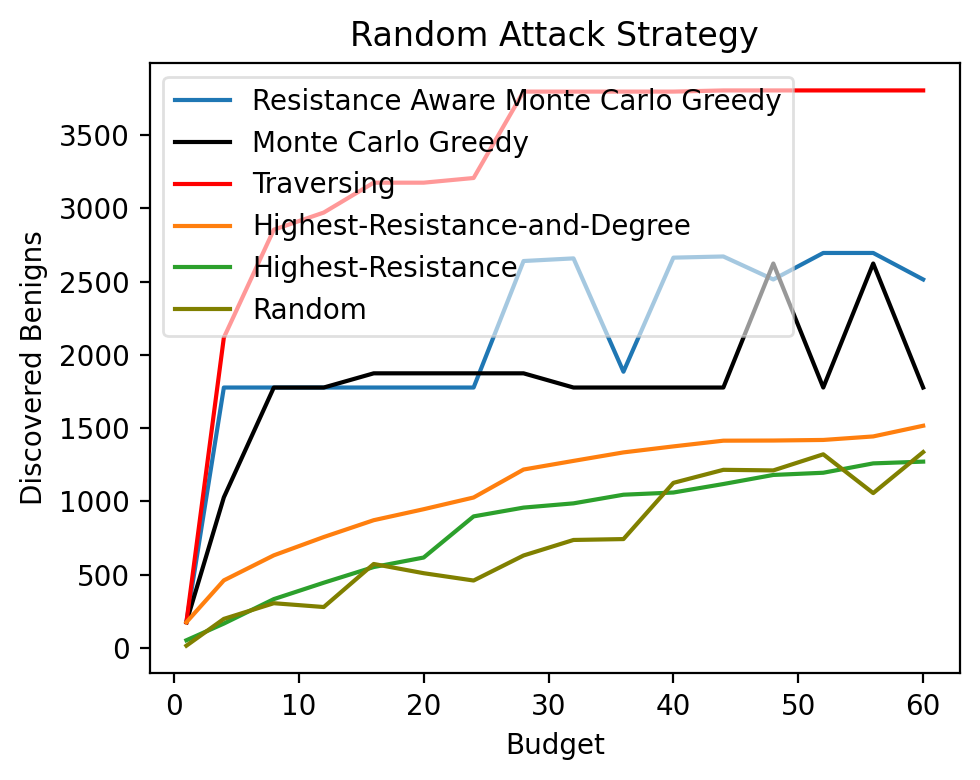}
    \end{subfigure}
    \hfill
    \begin{subfigure}[b]{0.32\textwidth}
        \centering
        \includegraphics[width=\linewidth]{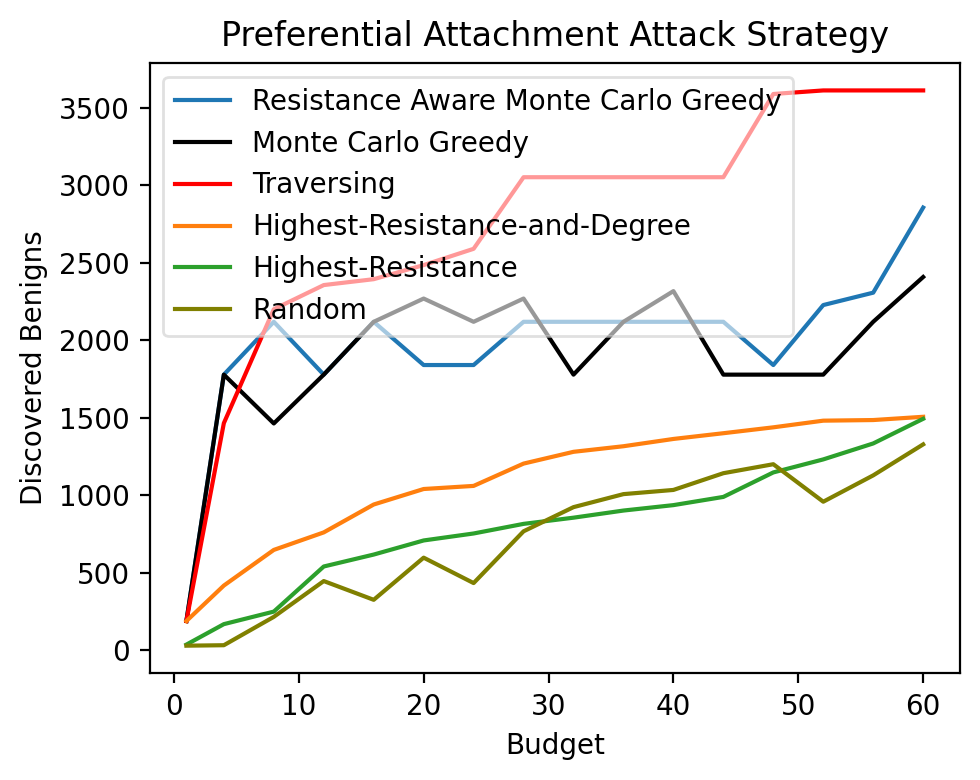}
    \end{subfigure}
    \hfill
    \begin{subfigure}[b]{0.32\textwidth}
        \centering
        \includegraphics[width=\linewidth]{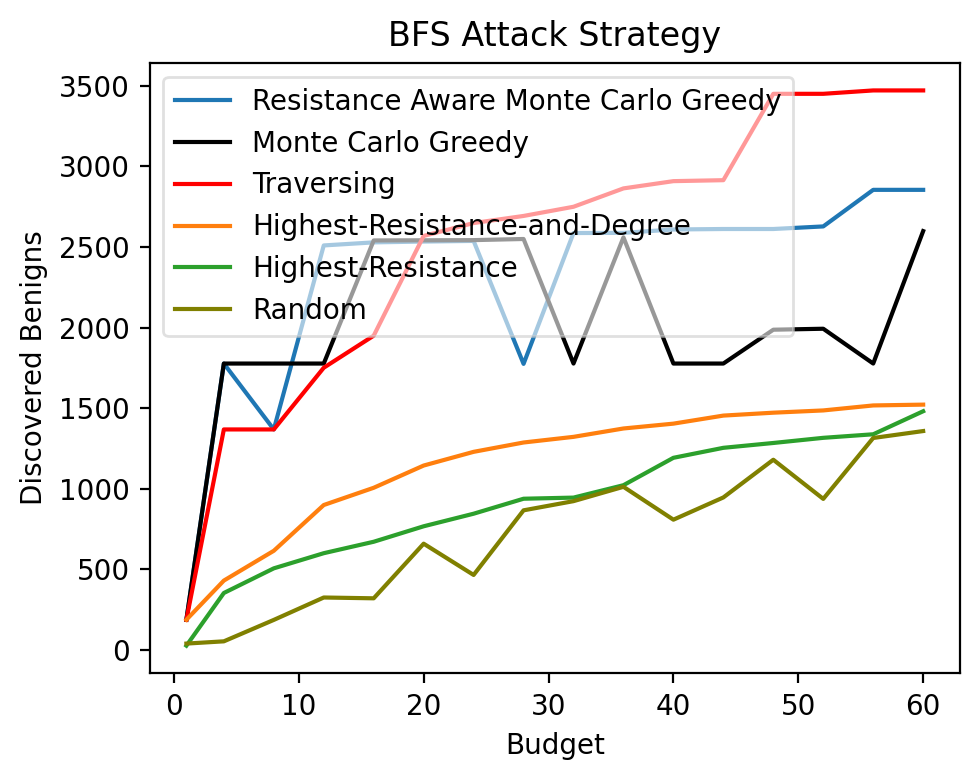}
    \end{subfigure}
    
    \caption{The number of discovered benigns by each algorithm when the budget ranges from $1$ to $k$ in the maximizing benigns problem on the \textit{Facebook} dataset and for different attack strategies. 
    }
    \Description{The figure shows the number of benigns discovered by various algorithms as the budget increases from $1$ to $k$ in the maximizing benigns problem on the Facebook dataset. The x-axis represents the budget size, while the y-axis represents the number of discovered benigns. Different plots are plotted for each algorithm under various attack strategies, indicating the effectiveness of each method as the budget changes.}
    \label{fig:mb_facebook}
\end{figure*}

\subsubsection{Potential Attack Edges}

We now evaluate the performance of our Proposed Algorithm~\ref{pseudocode:RPAE} for the potential attack edges (PAE) problem against the Random algorithm, which makes $k$ random choices. Figure~\ref{fig:PAE experiment for Facebook dataset} (top row) illustrates the number of PAEs found by each algorithm for a range of budgets. Our algorithm outperforms the Random algorithm. This is unsurprising since our algorithm is optimal, as we proved in Section~\ref{sec:pae}.

One natural question is what fraction of PAEs are, in fact, attack edges (edges from sybil to benign). We report this information in Figure~\ref{fig:PAE experiment for Facebook dataset} (second row). As one can observe, for our algorithm, around $20\%$ of found PAEs are attack edges for various budget choices. To get a better understanding of how good this performance is, we consider the Full-Knowledge algorithm. We suppose this algorithm knows all the values of resistance (note that this information is not available to our algorithm) and aims to greedily pick nodes which maximize the ratio of attack edges over PAEs. While for small budgets, the gap is large, as the budget grows, the ratio of attack edges over PAEs for our algorithm almost matches the Full-Knowledge algorithm, which is impressive considering that our algorithm doesn't know the exact values of resistance $r(\cdot)$.
For other datasets, similar results are provided in the Appendix~\ref{appendix:Additional PAE experiment reports}.

\subsection{Classification Algorithms}
\label{sec:classification}

Our algorithms for discovering benigns and potential attack edges can act as a preprocessing step for various sybil detection algorithms. We analyze the performance of three state-of-the-art detection algorithms, with and without such preprocessing step.

We consider the SybilWalk~\cite{jia2017random} and SybilSCAR~\cite{wang2017sybilscar} algorithms. We also analyze a node classifier algorithm based on logistic regression that we call SybilMetric~\cite{asghari2022using}. We first executed each of these algorithms without any preprocessing. We then applied the Traversing algorithm to maximize the number of benigns (our best algorithm for maximizing benigns based on our experiments in the previous section) before running the detection strategies. In the final setup, we also used our potential attack edges discovering algorithm on both the initial known benigns and the newly identified benigns from the first preprocessing step. Since detection algorithms usually rely on homophily property, removing (or reducing the weight of) attack edges could benefit them. Therefore, in the last setup, before running the detection algorithms, the weights for the discovered PAEs were adjusted to reduce their importance.

Furthermore, in all experiments, for both preprocessing phases, the budget is set to $1\%$ of benigns. We used the same setup as the prior work for the studied detection algorithms.

Our comparison results for the Facebook dataset are shown in Table~\ref{table:combined_facebook_results}. To measure the performance of the detection algorithms, we use AUC (Area Under the Curve), which is commonly used by the prior work since it is less affected by biases compared to accuracy. Based on these results, the performance of all three algorithms is improved after adding maximizing benigns outputs. Using discovered PAEs as a preprocessing step also enhances the AUC in some cases, but not always. This flags that one may need to incorporate the information regarding the PAEs directly into these algorithms instead of just reducing weights as we did to get a stronger boost in the performance. (This would be an interesting potential avenue for future research.) 
Thus, overall, our preprocessing steps, especially discovering benigns, can significantly boost the performance of detection strategies. For other datasets, similar results are provided in the Appendix~\ref{appendix:Other Experiments Classification}.

\begin{figure}[htbp]
    \centering
    \includegraphics[width=\linewidth]{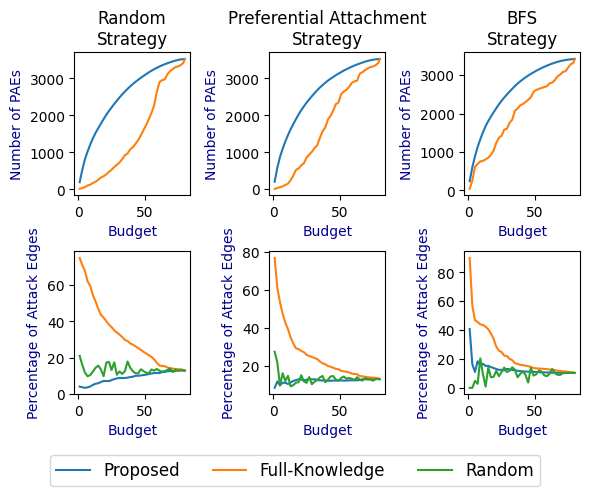}
    \caption{The number of discovered PAEs for different budgets (first row) and the percentage of attack edges relative to the number of discovered PAEs (second row) on the \textit{Facebook} dataset. Each column represents a different attack strategy.
    }
    \Description{The figure consists of two rows and three columns of plots. Each column corresponds to a distinct attack strategy. The first row displays the number of discovered PAEs (Protected Account Entities) for varying budget sizes, with the x-axis representing the budget and the y-axis representing the count of PAEs. The second row shows the percentage of attack edges relative to the number of discovered PAEs. The x-axis in this row represents different budget sizes, and the y-axis indicates the percentage of attack edges.}
    \label{fig:PAE experiment for Facebook dataset}
\end{figure}

It is also worth emphasizing that our preprocessing steps are not computationally demanding. More precisely, in most cases, the computational overhead added by the preprocessing step is negligible in comparison to the time required by the detection algorithm itself.

\section{Conclusion}

In this paper, we introduced novel attack strategies for synthesizing more realistic and generic datasets by introducing the concept of user resistance. We then considered two optimization problems where we leveraged resistance information to discover benigns and potential attack edges. We introduced several algorithms and theoretically analyzed their runtime and solution accuracy. We then showed the outcomes of these algorithms contain valuable information that can be used as a preprocessing step for detection algorithms. We conducted a large set of experiments that confirmed the positive impact of our preprocessing algorithms.

Future research could explore dynamic attack strategies to enhance the adaptability of sybil detection algorithms. Developing bias-robust algorithms is also essential, as current methods are vulnerable to dataset biases, especially when preprocessing reveals new benigns and skews the dataset.
In addition, developing algorithms for maximizing benigns problem with theoretical guarantees could be a potential avenue for future work.

\begin{table}[htbp]
    \centering
    \caption{Performance of SybilSCAR, SybilWalk, and SybilMetric with and without preprocessing on the \textit{Facebook} dataset. \textbf{Init} represents the setup without preprocessing. \textbf{MB} incorporates discovered benigns by the Traversing algorithm. \textbf{MB+PAE} incorporates both discovered benigns and PAEs.}
    \resizebox{\linewidth}{!}{
    \begin{tabular}{|c|c|c|c|c|}
        \hline
        \makecell{\textbf{Attack} \\ \textbf{Strategy}} & \textbf{Step} & \makecell{\textbf{SybilSCAR} \\ \textbf{AUC}} & \makecell{\textbf{SybilWalk} \\ \textbf{AUC}} & \makecell{\textbf{SybilMetric} \\ \textbf{AUC}} \\
        \midrule
        
        \multirow{3}{*}{Random} & Init & 0.924 & 0.966 & 1.00 \\
        \arrayrulecolor{gray}\cline{2-5}\arrayrulecolor{black}
        & MB & 0.988 & 0.998 & 1.00 \\
        \arrayrulecolor{gray}\cline{2-5}\arrayrulecolor{black}
        & MB+PAE & 0.988 & 0.998 & 0.99 \\
        \cline{1-5}
        \multirow{3}{*}{BA} & Init & 0.876 & 0.929 & 1.00 \\
        \arrayrulecolor{gray}\cline{2-5}\arrayrulecolor{black}
        & MB & 0.954 & 0.972 & 1.00 \\
        \arrayrulecolor{gray}\cline{2-5}\arrayrulecolor{black}
        & MB+PAE & 0.944 & 0.972 & 1.00 \\
        \cline{1-5}
        \multirow{3}{*}{BFS} & Init & 0.986 & 0.985 & 0.97 \\
        \arrayrulecolor{gray}\cline{2-5}\arrayrulecolor{black}
        & MB & 0.995 & 0.996 & 0.99 \\
        \arrayrulecolor{gray}\cline{2-5}\arrayrulecolor{black}
        & MB+PAE & 0.991 & 0.997 & 1.00 \\
        \hline
    \end{tabular}
    }
    \label{table:combined_facebook_results}
\end{table}





\balance
\bibliographystyle{ACM-Reference-Format} 



\appendix

\counterwithin{figure}{section}
\counterwithin{table}{section}
\counterwithin{algorithm}{section}

\section{Pseudocode of Attack Strategies}
\label{appendix:Algorithms of Attack Strategies}

The pseudocode for the Random, Preferential Attachment, and BFS attack strategies are provided in Algorithms~\ref{alg:Firststrategy},~\ref{alg:Secondstrategy}, and~\ref{alg:Thirdstrategy}, respectively.


\captionsetup[algorithm]{labelformat=empty} 
\renewcommand{\thealgorithm}{\thesection.\arabic{algorithm}}  

\begin{algorithm}[H]
\caption{\textbf{Attack Strategy \thealgorithm} Random}
\raggedright  \textbf{Input} $G=(V,E),S,B,AE(\cdot),r(\cdot),c$
\\
\raggedright  \textbf{Output} Generated graph $G'$
\begin{algorithmic}[1]
\For{each node $s_i$ in set $S$}
    \State $chosen\_nodes \gets \emptyset$
    \For{$i \leftarrow 1$ to $c \cdot AE(i)$}
        \State choose a node $u$ from $(B \setminus chosen\_nodes)$ randomly
        \State $chosen\_nodes \gets chosen\_nodes \cup \{u\}$
        
        \If{ $r(u) = 0$ }
            
            \State $E \gets E \cup \{(s_i, u)\}$
            \State with probability $\frac{1}{2}$: $E \gets E \cup \{(u,s_i)\}$
            
        \EndIf
    \EndFor
\EndFor

\State \Return $G'=(V,E)$

\end{algorithmic}
\label{alg:Firststrategy}
\end{algorithm}

\captionsetup[algorithm]{labelformat=empty} 
\renewcommand{\thealgorithm}{\thesection.\arabic{algorithm}}  

\begin{algorithm}[H]
\caption{\textbf{Attack Strategy \thealgorithm} Preferential Attachment (PreAt)}
\label{attackstrategyba}
\raggedright  \textbf{Input} $G=(V,E),S,B,AE(\cdot),r(\cdot),c$
\\
\raggedright  \textbf{Output} Generated graph $G'$
\begin{algorithmic}[1]
\Function{modified\_BA}{$G=(V,E),B,S$}
    \For{each node $u$ in $B$}
        \State $P_1(u) \gets \text{BA probability of in-degree}~u \text{ from } B$
        \State $P_2(u) \gets$ BA probability of in-degree $u$ from $S$
        \State $P(u) \gets \frac{P_1(u) + P_2(u)}{2}$
    \EndFor
    \State \Return $P$
\EndFunction
\Statex
\For{each node $s_i$ in set $S$}
    \State $chosen\_nodes \gets \emptyset$

    \State $P \gets \Call{modified\_BA}{}$
    \For{$i \leftarrow 1$ to $c \cdot AE(i)$}
        
        \State $u \gets$ a random node from $(B \setminus chosen\_nodes)$ following the probability vector $P$
        \State $chosen\_nodes \gets chosen\_nodes \cup \{u\}$

        \If{ $r(u) = 0$ }
            
            \State $E \gets E \cup \{(s_i, u)\}$
            \State with probability $\frac{1}{2}$: $E \gets E \cup \{(u,s_i)\}$
            
        \EndIf        
    \EndFor
\EndFor

\State \Return $G'=(V,E)$

\end{algorithmic}
\label{alg:Secondstrategy}
\end{algorithm}

\captionsetup[algorithm]{labelformat=empty} 
\renewcommand{\thealgorithm}{\thesection.\arabic{algorithm}}  

\begin{algorithm}[H]
\caption{\textbf{Attack Strategy \thealgorithm} BFS}
\raggedright  \textbf{Input} $G=(V,E),S,B,dual(\cdot),r(\cdot),c$
\\
\raggedright  \textbf{Output} Generated graph $G'$ 
\begin{algorithmic}[1]
\For{each node $s_i$ in set $S$}
    \State $chosen\_nodes \gets \emptyset$
    \State $BFS\_streamer \gets$ BFS starts from node $dual(s_i)$ 
    \State $P \gets \Call{modified\_BA}{G=(V,E),B,S}$ from Algorithm~\ref{attackstrategyba}

    \For{$i \leftarrow 1$ to $c \cdot AE(i)$}
    
        \State \textbf{if} any node is available in $BFS\_streamer$:
            \State \quad $u \gets$ next node from $BFS\_streamer$
        \State \textbf{else}:
            \State \quad $u \gets$ a random node from $(B \setminus chosen\_nodes)$ following the probability vector $P$ 

        \State $chosen\_nodes \gets chosen\_nodes \cup \{u\}$
        \If{ $r(u) = 0$ }
            \State $E \gets E \cup \{(s_i, u)\}$
            \State with probability $\frac{1}{2}$: $E \gets E \cup \{(u,s_i)\}$
        \EndIf
    \EndFor
\EndFor

\State \Return $G'=(V,E)$

\end{algorithmic}
\label{alg:Thirdstrategy}
\end{algorithm}
\captionsetup[algorithm]{labelformat=default} 

\section{Example for Transformer}
\label{appendix:Connection Between Optimal Solutions example}

See Figure~\ref{fig:MCtoMB} for an example of the transformer used in the proof of hardness for the maximizing benigns problem.

The transformer converts $I = \langle W,Q,k \rangle$, which is an instance of MC, into $I'=\langle G,B,S,k,p_r(\cdot) \rangle$, which is an instance of MB.
In this example 
$V_W = \{w_1, w_2, w_3, w_4, w_5\}$, 
$V_Q=\{q_1,q_2,q_3,q_4\}$, and $q_1=\{w_1,w_3\}$, $q_2=\{w_5\}$, $q_3=\{w_2, w_5\}$, $q_4=\{w_3,w_4\}$. In transformation, we have $B=V_Q$ and $S=\emptyset$. We keep the value of $k$ and set $p_r(v) = 1 $ for all nodes based on transformer rules.

\begin{figure}[H]
    \centering
    \includegraphics[width=0.85\linewidth]{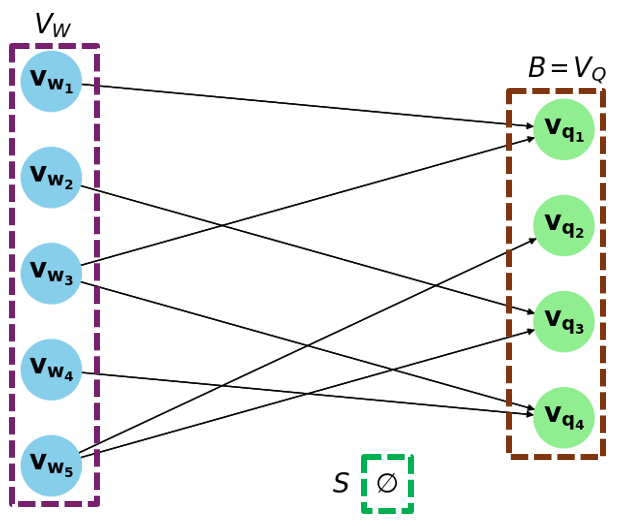}
    \caption{The graph outcome of transformer where $S=\emptyset$ and $V_Q=B$.}
    \Description{The figure illustrates a graph transformation process. On the left, two vertices are shown: $V_Q$ and $V_W$. Each $v_{q_i}$ in $V_Q$ is connected to one or more $v_{w_j}$ in $V_W$. In the figure, one class of nodes is assigned to $B$.}
    \label{fig:MCtoMB}
\end{figure}

\section{Simple Greedy Algorithm}
\label{appendix:mbgreedy}
Algorithm~\ref{pseudocode:greedyMB3242} provides the description of the simple greedy algorithm for the maximizing benigns problem. In this algorithm, we write $f(A \cup \{u\})$ for simplicity, but, to be precise, it should be expressed as $f(A \cup \{u\}, B, S, p_r(\cdot))$, but the other parameters are clear from the context.

\counterwithin{algorithm}{section}

\begin{algorithm}
\caption{Simple Greedy Algorithm}
\label{pseudocode:greedyMB3242}
\raggedright \textbf{Input} $G=(V,E),B, S, \text{budget }k,p_r(\cdot)$\\
\raggedright \textbf{Output} $reveal\_set\ A$

\begin{algorithmic}[1]
\State $A \gets \emptyset$  
\For{$k$ iterations}   
    \State $u \gets \arg\max_{u \in V \setminus A} f(A \cup \{u\})$  
    \State $A \gets A \cup \{u\}$  
\EndFor
\State \Return $A$

\end{algorithmic}
\end{algorithm}

\section{Proof for Non-submodularity}
\label{Appendix:Proof for Nonsubmodularity}

To prove that the function $f(\cdot)$ is not submodular, it suffices to provide an example such that
$v \in V,
A_1 \subseteq A_2$ and
\begin{equation*}
    f(A_2\cup \{v\}) - f(A_2)
    >
    f(A_1\cup \{v\}) - f(A_1)  
\end{equation*}

Consider the graph $G$ given in Figure~\ref{fig:notsubmodular}, let's define
$p_r(u_i)=\frac{1}{2}$ for all nodes $u_i$,
$S = \emptyset$
and
$B = \{u_1\}$. We define:
$A_1 = \{u_1\}$,
$A_2 = \{u_1,u_2\}$ and
$v = u_3$.

Let us compute each of the four entities one by one.
We have $f(A_1) = \frac{1}{2}$ because there is a 1/2 probability that the node $u_1$ is revealed to be resistant, which results in discovering that the node $u_2$ is benign. In addition, 
$f(A_1 \cup \{v\}) = f(A_1 \cup \{u_3\}) = \frac{1}{2}$ because revealing node $u_3$ cannot result in the discovery of any new benign. This is true because of Observation~\ref{obs:pathtoB_toDescovered}.
Thus, $f(A_1 \cup \{v\})-f(A_1)=0$.

Furthermore, $f(A_2) = \frac{3}{4}$ because there is a 1/2 probability of discovering \( u_2 \) to be benign by revealing \( u_1 \) to be resistant, and then a \( 1/2 \times 1/2 = 1/4 \) probability of discovering \( u_3 \) if both \( u_1 \) and \( u_2 \) turn out to be resistant.
We also have $f(A_2 \cup \{v\}) = \frac{3}{4} + \frac{n-3}{8}$,
because we have the same \(\frac{3}{4}\) as before, and with a probability of \(\frac{1}{8}\), \( n-3 \) new benigns will be discovered. Thus, we have $f(A_2 \cup \{v\})-f(A_2)=\frac{n-3}{8}$. Since $\frac{n-3}{8}>0$, for any $n\ge 4$, we can conclude the proof. 

\begin{figure}[H]
    \centering
    \includegraphics[width=0.85\linewidth]{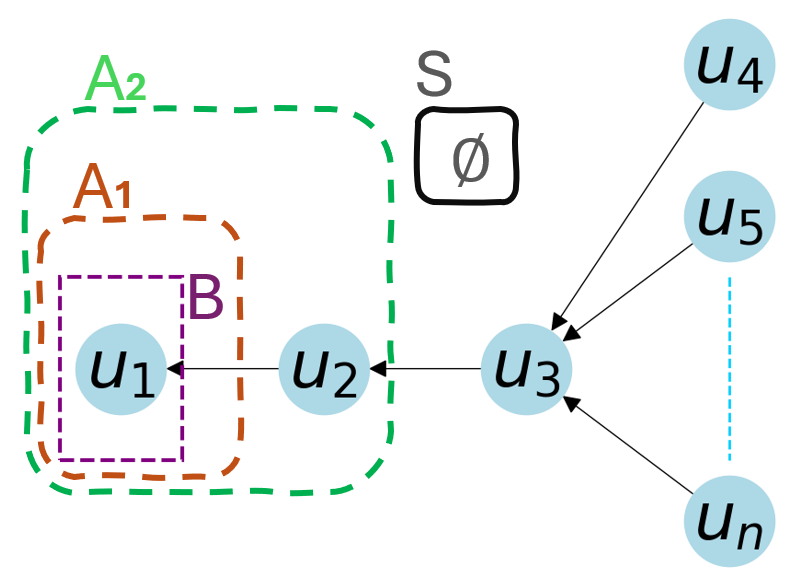}
    \caption{An example to show $f(\cdot)$ is not submodular.}
    \Description{Six nodes connected by directed edges and $u_1$ as a start point, and $u_2$ is its first level neighbor. $u_3$ is in the next level, and other nodes are connected to $u_3$. This figure is for checking the submodularity.}
    \label{fig:notsubmodular}
\end{figure}

\section{BFS-based Discovery Algorithm}
\label{appendix:BFS-based Discovery Algorithm}

Below, we provide the Algorithm~\ref{algorithm:BFS-basedDiscovery} for the BFS-based algorithm, which is used in line 11 of Algorithm~\ref{algorithm:montecarlo}.

\begin{algorithm}[H]
\caption{BFS-based Discovery Algorithm}
\label{algorithm:BFS-basedDiscovery}
\raggedright  \textbf{Input} $G=(V,E),\text{reveal set }A,B,S,r(\cdot)$
\\
\raggedright  \textbf{Output} $\text{size of discovered set}$
\begin{algorithmic}[1]
\State $\text{discovered} \gets \emptyset$
\State $\text{visited} \gets \emptyset$
\For{$v \in A \cap B$}
    \State $\text{bfs\_queue} \gets \text{empty queue}$
    \If{$r(v) = 1$}
        \State $\text{bfs\_queue.add}(v)$
    \EndIf
    \While{$\text{bfs\_queue is not empty}$}
        \State $w \gets \text{bfs\_queue.pop}()$
        \State $visited \gets visited \cup \{w\}$
        \State $\text{discovered} \gets \text{discovered} \cup (\Gamma_{\text{in}}(w) \setminus (B \cap S))$
        \For{$u \in (\Gamma_{\text{in}}(w) \cap A)\setminus visited$}
            \If{$r(u) = 1$}
                \State $\text{bfs\_queue.add}(u)$
            \EndIf
        \EndFor
    \EndWhile
\EndFor
\State \Return $|discovered|$
\end{algorithmic}
\end{algorithm}

\section{\texorpdfstring{Bound on $R$ in Monte Carlo Greedy Algorithm}{Bound on R in Monte Carlo Greedy Algorithm}}
\label{appendix:hoeffding Proof of bound}

\begin{theorem}[Hoeffding's inequality~\cite{hoeffding1994probability}]
\label{theorem:hoeffding}
Consider independent random variables $X_1, X_2, \ldots, X_R$ such that $X_i \in [a,b]$, $\bar{X}$ is the mean of $R$ random variables and $\mu$ is the expected value of $\bar{X}$. 
Then, $\mathbb{P}(\bar{X} - \mu \ge \epsilon) \le \exp(\frac{-2R\epsilon^2}{(b-a)^2})$, where $\epsilon$ is the error margin.
\end{theorem}

To gain confidence $\alpha$ using the above lemma, one requires
$R \ge 
\frac{(b-a)^2 \ln (\frac{1}{1-\alpha})}{2\epsilon^2}$.
%
%
According to the inequality, we have
\[
\mathbb{P}(\bar{X} - \mu <  \epsilon) \ge 1 - \exp(\frac{-2R\epsilon^2}{(b-a)^2}).
\]
By setting $\mathbb{P}(\bar{X} - \mu <  \epsilon)  \ge \alpha$ we will have  
\[
 1 - \exp(\frac{-2R\epsilon^2}{(b-a)^2}) \ge \alpha, 
\]
Which means
\[
 \ln(1-\alpha) \ge  \frac{-2R\epsilon^2}{(b-a)^2}, 
\]
which yields 
\[
 R \ge \frac{(b-a)^2 \ln (\frac{1}{1-\alpha})}{2\epsilon^2}. 
\]

Suppose that there are $R$ Monte Carlo experiments, and the $i$-th experiment estimates the number of discovered newly benigns based on each $A$ as $X_i$. The above calculated bound on $R$ in this problem will be $R \ge \frac{k^2\Delta_{\text{in}}^2 \ln(\frac{1}{1-\alpha})}{2\epsilon^2}$
if we set, $a = 0$ and $b \le |A|\cdot\Delta_{\text{in}} \le k\Delta_{\text{in}}$. 

\section{Run Time of Monte Carlo Greedy Algorithm}
\label{appendix:complexity discussion}
The complexity of the Monte Carlo Greedy algorithm (described in Algorithm~\ref{algorithm:montecarlo}) is determined by the $k$ iterations to each time identify the best node among at most $n$ candidates. For each node, the computational cost is defined by the \texttt{EST} function. This function involves assigning $r(v)$ for each $v \in A$ based on $p_r(v)$, which incurs a cost of $O(|A|)$, which is in $O(n)$. Subsequently, it computes the number of discovered benigns using the algorithm described in Appendix~\ref{appendix:BFS-based Discovery Algorithm}, which has a complexity of $O(n+m)$. This algorithm runs a BFS with different roots, but as we have a common list of visited nodes in the algorithm, the complexity is still $O(n+m)$. Therefore, the overall complexity is given by $O(k(n + n \cdot R(n+m)))$, which simplifies to $O(k \cdot R \cdot (n^2 + nm))$.

\section{Example for Traversing Algorithm}
\label{appendix:TraversingResistanceDegreeAware}

Here, we present an example for Algorithm~\ref{algorithm:TraversingResistanceDegreeAware}, visualized in Figures~\ref{fig:Datastructureforalgorithm} and~\ref{fig:Graphschemaforalgorithm}, to illustrate its execution and the associated data structure. 

We have
$B = \{a,b\}$, $S=\emptyset$ and $N = \{a,b\}$. Initially, 
$\hat{\Gamma}_{\text{in}}(a) = \{d,f\}$,
$\hat{\Gamma}_{\text{in}}(e) = \{c,f\}$, 
$\hat{\Gamma}_{\text{in}}(c) = \{d\}$,
$\gamma_{\text{in}}(e) = 2$, and
$\gamma_{\text{in}}(c) = 1$.
In addition, 
$\Gamma_{\text{out}}(d) = \{c,a\}$, and
$\Gamma_{\text{out}}(f) = \{a,b,e\}$.
Suppose Algorithm~\ref{algorithm:TraversingResistanceDegreeAware} selects $a$ and $r(a)=1$. Then, in line $9$, we will have $N=\{a,b\}\setminus \{a\}=\{b\}$. From $a$'s in-neighbors table, we walk to its in-neighbors $d$ and $f$. From each of them, by purple links, move to their out-neighbors. For example, from the node $f$ on the left side, with the help of a purple link, we move to $f$ cell in the right table, and then, for its connected nodes, include $e$ and $b$ ($a$ is traversed before, and we do not care about that), we will use orange and blue links. For example, with the orange link, we go to $f$, which is connected to $e$ in the right table. We will remove that node and the connection there. Then, we can go to the related counter and decrease it once. Then we have to move backward and do the same for cell $d$ in the left table, and from its neighbor $c$, with a pink link, go to node $d$ on the right side and do the same. In the output of this step, we have:
$\hat{\Gamma}_{\text{in}}(e) = \{c\}$, 
$\hat{\Gamma}_{\text{in}}(c) = \{\}$,
$\gamma_{\text{in}}(e) = 1$, and
$\gamma_{\text{in}}(c) = 0$. Please note that we avoid drawing all links in the figure, as it would make it really unreadable.

\begin{figure*}[ht!]
    \centering
    \begin{minipage}{0.77\linewidth}
        \centering
        \includegraphics[width=\linewidth]{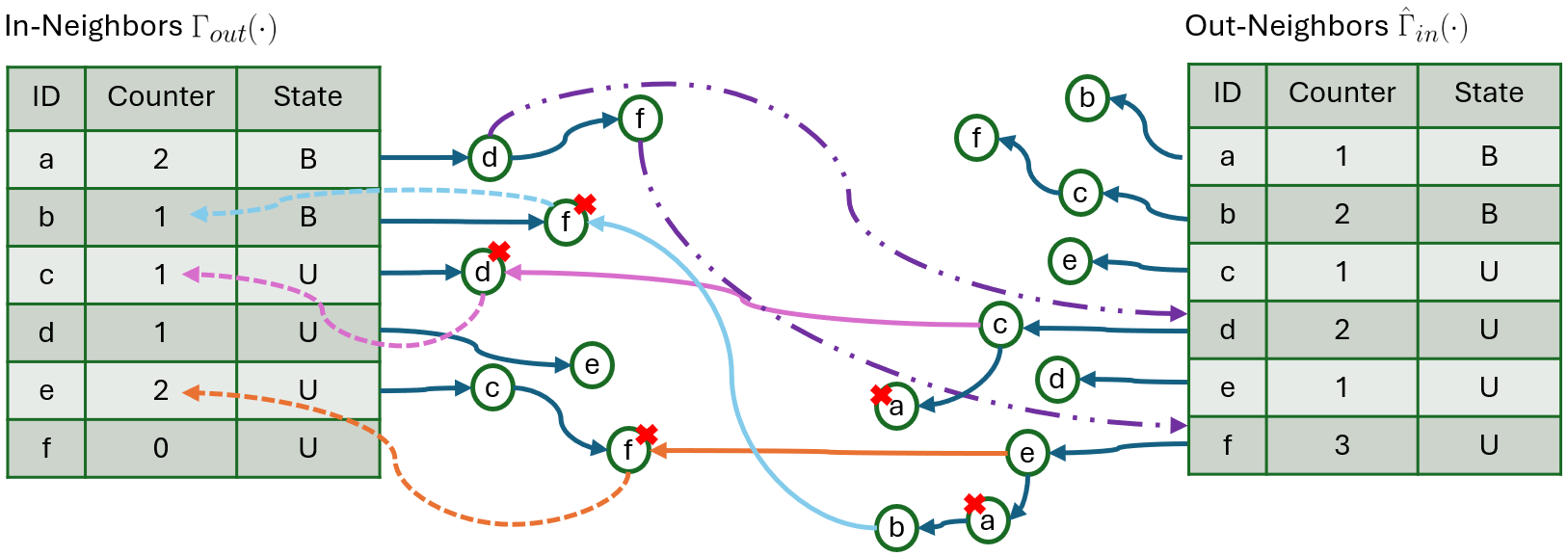}
        \caption{Data structure schema used to update neighborhoods 
        in Algorithm~\ref{algorithm:TraversingResistanceDegreeAware}. First, starting from $a$ with purple lines, we can reach $d$ and $f$ on the right side table, and then based on three lines shown by blue, orange, and pink lines, we can apply the updates.}
        \Description{The figure shows a schema illustrating how neighborhoods are updated in a graph. Starting from node $a$, it connects to nodes $d$ and $f$. Three updates are represented by blue, orange, and pink lines, indicating changes in node connections. Two tables display the in-neighbors (left) and out-neighbors (right) for each node, reflecting the step-by-step updates.}
        \label{fig:Datastructureforalgorithm}
    \end{minipage}%
    \hfill
    \begin{minipage}{0.22\linewidth}
        \centering
        \includegraphics[width=\linewidth]{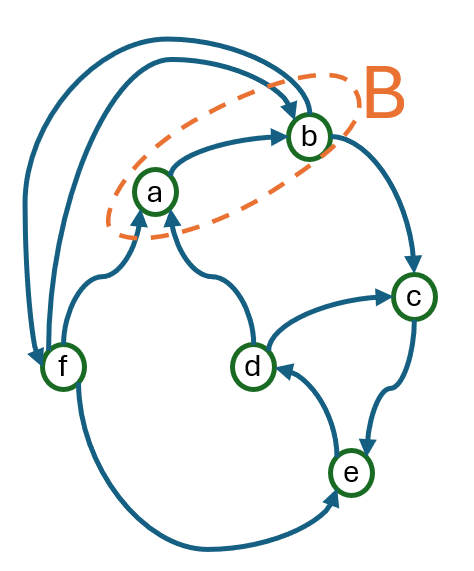}
        \caption{Example graph for describing the data structure used in Algorithm~\ref{algorithm:TraversingResistanceDegreeAware}. Nodes $a$ and $b$ form the benign set.}
        \Description{The figure includes nodes labeled $\{a, b, c, d, e, f\}$. An orange boundary highlights the benign region and consists of nodes $\{a, b\}$.}
        \label{fig:Graphschemaforalgorithm}
    \end{minipage}
\end{figure*}

\section{Real-World Datasets Statistics Reports}
\label{appendix:Datasets Real world Statistics reports}

Table~\ref{table:real_word_stat} contains some statistics about the real-world datasets we used.

\begin{table}[ht!]
\centering
\caption{Some statistics on real-world networks studied.}
\label{table:real_word_stat}
\begin{tabular}{|l|c|c|c|c|}
\hline
Network & Facebook & Twitter & LastFM & Pokec \\
\midrule
$|$Nodes$|$ & 4,039 & 10,000 & 7,624 & 10,000 \\ \arrayrulecolor{gray}\hline\arrayrulecolor{black}
$|$Undirected Edges$|$ & 88,234 & 0 & 27,806 & 0 \\ \arrayrulecolor{gray}\hline\arrayrulecolor{black}
$|$Directed Edges$|$ & 0 & 350,600 & 0 & 94,066 \\ \arrayrulecolor{gray}\hline\arrayrulecolor{black}
Avg $d_{\text{in}}$ & 43.69 & 35.06 & 7.29 & 9.4 \\ \arrayrulecolor{gray}\hline\arrayrulecolor{black}
Avg $d_{\text{out}}$ & 43.69 & 35.06 & 7.29 & 9.4 \\ \arrayrulecolor{gray}\hline\arrayrulecolor{black}
Is directed? & - & $\checkmark$ & - & $\checkmark$ \\
\hline
\end{tabular}
\end{table}

\section{Datasets Additional Statistics Reports after Attacks}
\label{appendix:Additional Datasets after Attacks Statistics reports}

For our studied datasets, the statistics after Random, PreAt (Preferential Attachment), and BFS attacks are summarized in Table~\ref{table:CombinedStatistics}. The number of attack edges and the number of edges that are reverse of attacks are shown by $|$Attacks$|$ and $|$Rev. Attacks$|$.

\begin{table*}[ht!]
\centering
\caption{Statistics for \textit{Facebook}, \textit{Twitter}, \textit{LastFM}, and \textit{Pokec} datasets after Random, PreAt (Preferential Attachment), and BFS attacks. $s\rightarrow s$ and $b\rightarrow b$ indicate sybil-to-sybil and benign-to-benign edges. The next four rows represent the average number of incoming/outgoing edges for sybils and benigns. $r(\cdot)$ and $p_r(\cdot)$ represent resistance and probability of resistance.}
\label{table:CombinedStatistics}
\resizebox{\textwidth}{!}{
\begin{tabular}{||l||c|c|c||c|c|c||c|c|c||c|c|c||}
\hline
\multirow{2}{*}{Statistic} & \multicolumn{3}{c||}{Facebook+} & \multicolumn{3}{c||}{Twitter+} & \multicolumn{3}{c||}{LastFM+} & \multicolumn{3}{c||}{Pokec+} \\ \cline{2-13}
                           & Random & PreAt & BFS & Random & PreAt & BFS & Random & PreAt & BFS & Random & PreAt & BFS \\ \midrule
$|Edges|$                  & 203459 & 222080 & 202870 & 361687 & 393084 & 351979 & 66259 & 67546 & 65402 & 128875 & 142686 & 125539 \\ \hline
$|Attacks|$                & 13003 & 25396 & 12521 & 25073  & 45968  & 18619  & 3063  & 3919  & 2530  & 16870  & 26090  & 14683  \\ \hline
$|Rev.\ Attacks|$          & 6450   & 12678  & 6343   & 12585  & 23087  & 9331   & 1552  & 1983  & 1228  & 8494   & 13085  & 7345   \\ \hline
$|s\rightarrow s|$         & 7538   & 7538   & 7538   & 16484  & 16484  & 16484  & 6032  & 6032  & 6032  & 9446   & 9446   & 9446   \\ \hline
$|b\rightarrow b|$         & 176468 & 176468 & 176468 & 307545 & 307545 & 307545 & 55612 & 55612 & 55612 & 94065  & 94065  & 94065  \\ \hline
Avg $d_{\text{in}}$ Sybils    & 34.71  & 50.16  & 34.44  & 29.069 & 39.571 & 25.815 & 9.953 & 10.518 & 9.528  & 17.94  & 22.531 & 16.791 \\ \hline
Avg $d_{\text{out}}$ Sybils   & 50.97  & 81.72  & 49.77  & 41.557 & 62.452 & 35.103 & 11.936 & 13.059 & 11.236 & 26.316 & 35.536 & 24.129 \\ \hline
Avg $d_{\text{in}}$ Benigns   & 46.91  & 49.98  & 46.79  & 33.262 & 35.351 & 32.616 & 7.696 & 7.808  & 7.626  & 11.094 & 12.016 & 10.875 \\ \hline
Avg $d_{\text{out}}$ Benigns  & 45.29  & 46.83  & 45.26  & 32.013 & 33.063 & 31.688 & 7.498 & 7.554  & 7.455  & 10.256 & 10.715 & 10.141 \\ \hline
Avg $r(\cdot)$         & 0.7658 & 0.7658 & 0.7658 & 0.7519 & 0.7519 & 0.7519 & 0.7542 & 0.7542 & 0.7542 & 0.7464 & 0.7464 & 0.7464 \\ \hline
Avg $p_r(\cdot)$       & 0.6225 & 0.6309 & 0.6431 & 0.6216 & 0.6248 & 0.6222 & 0.6278 & 0.6310 & 0.6287 & 0.6214 & 0.6194 & 0.6268 \\ \hline
\end{tabular}
}
\end{table*}

\section{Resistance Aware Monte Carlo Greedy Algorithm}
\label{Appendix:montecarloresistanceaware}

Algorithm~\ref{algorithm:montecarloresistanceaware} describes the Resistance Aware Monte Carlo Greedy algorithm. This is the variant of the Monte Carlo Greedy algorithm, which reveals the resistance of each node once it is selected and leverages that information for the next choices.

\begin{algorithm}[H]
\caption{Resistance Aware Monte Carlo Greedy Algorithm}
\label{algorithm:montecarloresistanceaware}
\raggedright  \textbf{Input} $G=(V,E),B, S, \text{budget }k,p_r(\cdot),\epsilon, \alpha$
\\
\raggedright  \textbf{Output} $reveal\_set\ A$
\\
\begin{algorithmic}[1]
\State Initialize $A \gets \emptyset$
\State Initialize $A' \gets \emptyset$
\For{$i = 1$ to $k$}
    \State $v \gets \arg\max_{w \in V\setminus A} \Call{Est}{A' \cup \{w\}, p_r(\cdot),B,S,\epsilon, \alpha}$
    \State $A \gets A \cup \{v\}$
    \If {$r(v) = 1$} 
        \State $A' \gets A' \cup \{v\}$
    \EndIf
\EndFor
\State \Return $A$

\item[]
\Function{Est}{$A,p_r(\cdot),B,S,\epsilon,\alpha$}
    \State Initialize $count \gets 0$
    \State $R \gets \left\lceil  \frac{k^2\Delta_{\text{in}}^2 \ln(\frac{1}{1-\alpha})}{2\epsilon^2} \right\rceil$
    \For{$R$ iterations} 
        \State Assign resistance $r(v)$ for each $v \in A$ based on $p_r(v)$
        \State $new\_benign \gets$ number of discovered benigns based on $B$, $S$, and assigned resistances.
        \State $count \gets count + new\_benign$
    \EndFor
    \State \Return $count / R$  
\EndFunction

\end{algorithmic}
\end{algorithm}



        


\section{Additional Experiments for Maximizing Benigns}
\label{appendix:Additional MB experiment reports}

For the experiments of the maximizing benigns problem, the results on Twitter, LastFM, and Pokec datasets are provided in Figures~\ref{fig:mb_twitter}, \ref{fig:mb_lastfm}, and~\ref{fig:mb_pokec}.

\begin{figure*}[ht!]
    \centering
    
    \includegraphics[width=0.32\linewidth]{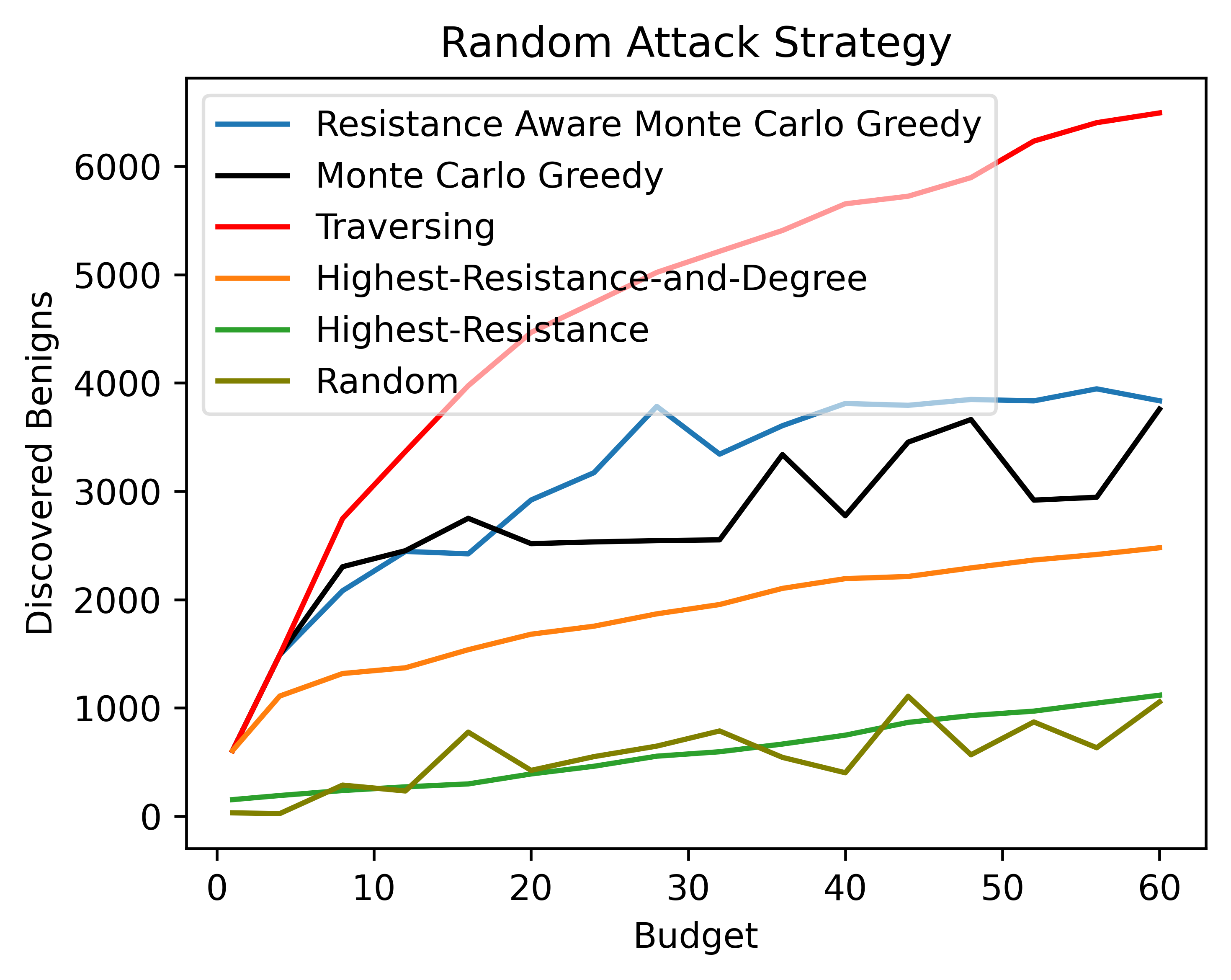}
    \includegraphics[width=0.32\linewidth]{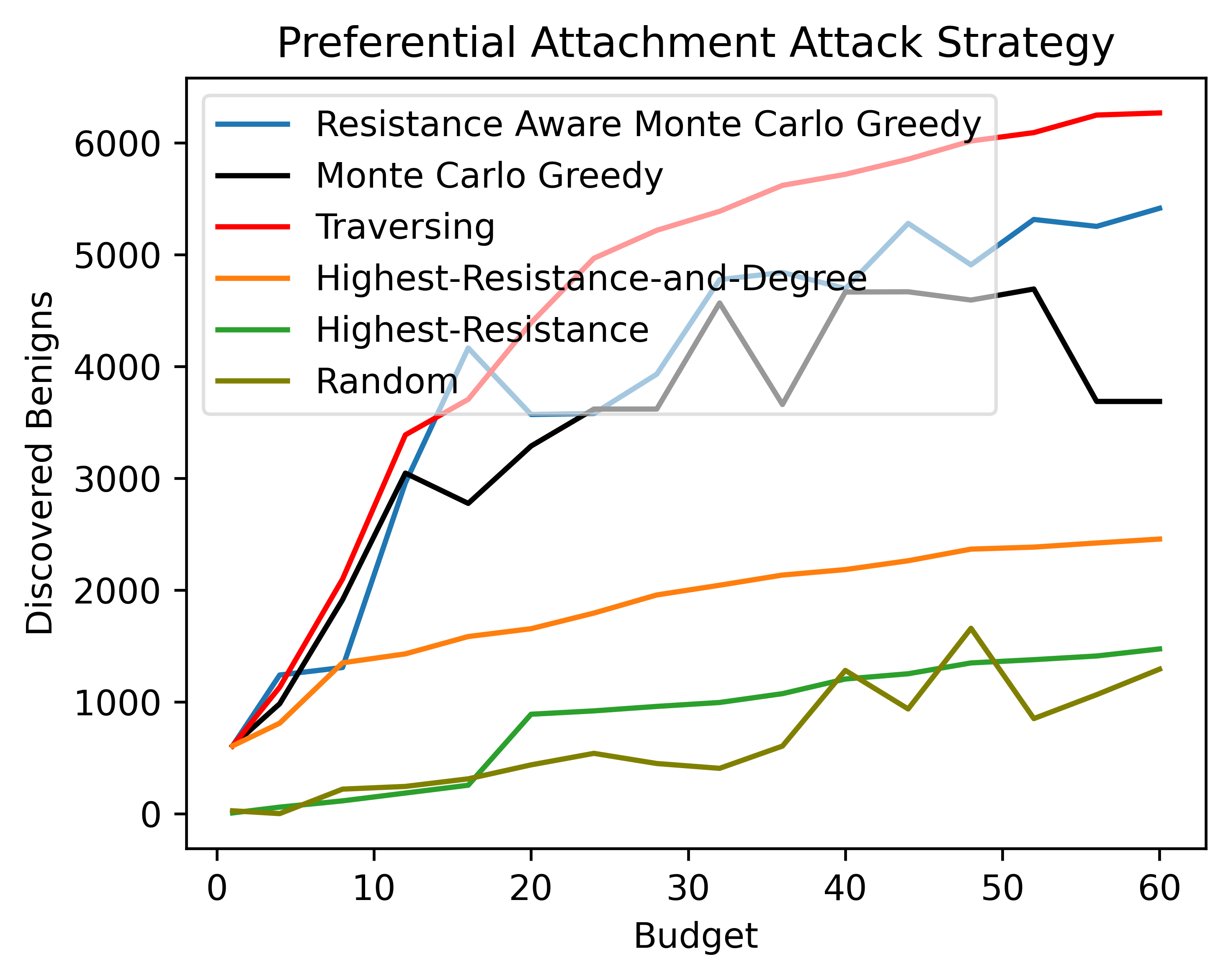}
    \includegraphics[width=0.32\linewidth]{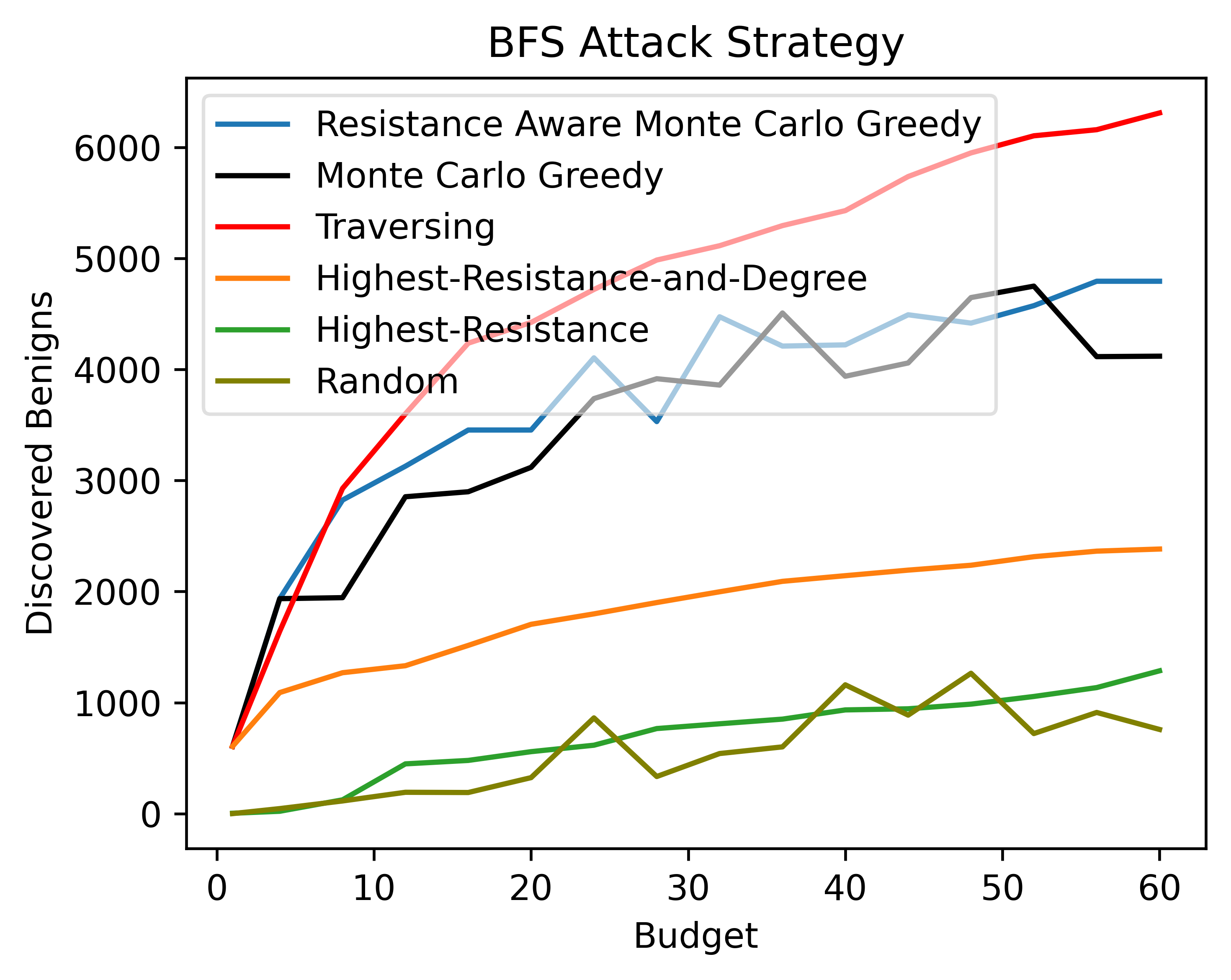}
    \caption{The number of discovered benigns by each algorithm when the budget ranges from $1$ to $k$ in the maximizing benigns problem on the \textit{Twitter} dataset and for different attack strategies.}
    \label{fig:mb_twitter}

    \includegraphics[width=0.32\linewidth]{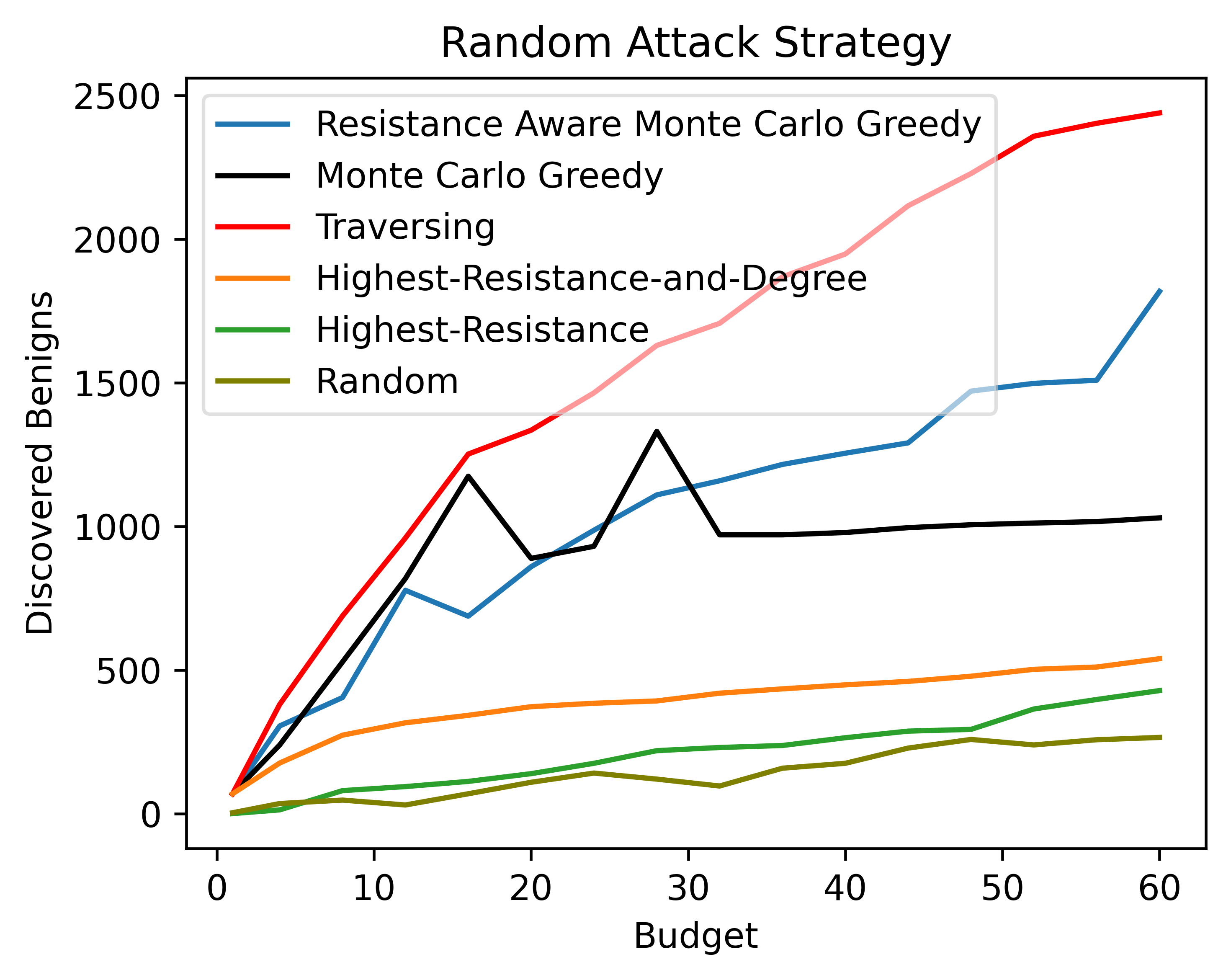}
    \includegraphics[width=0.32\linewidth]{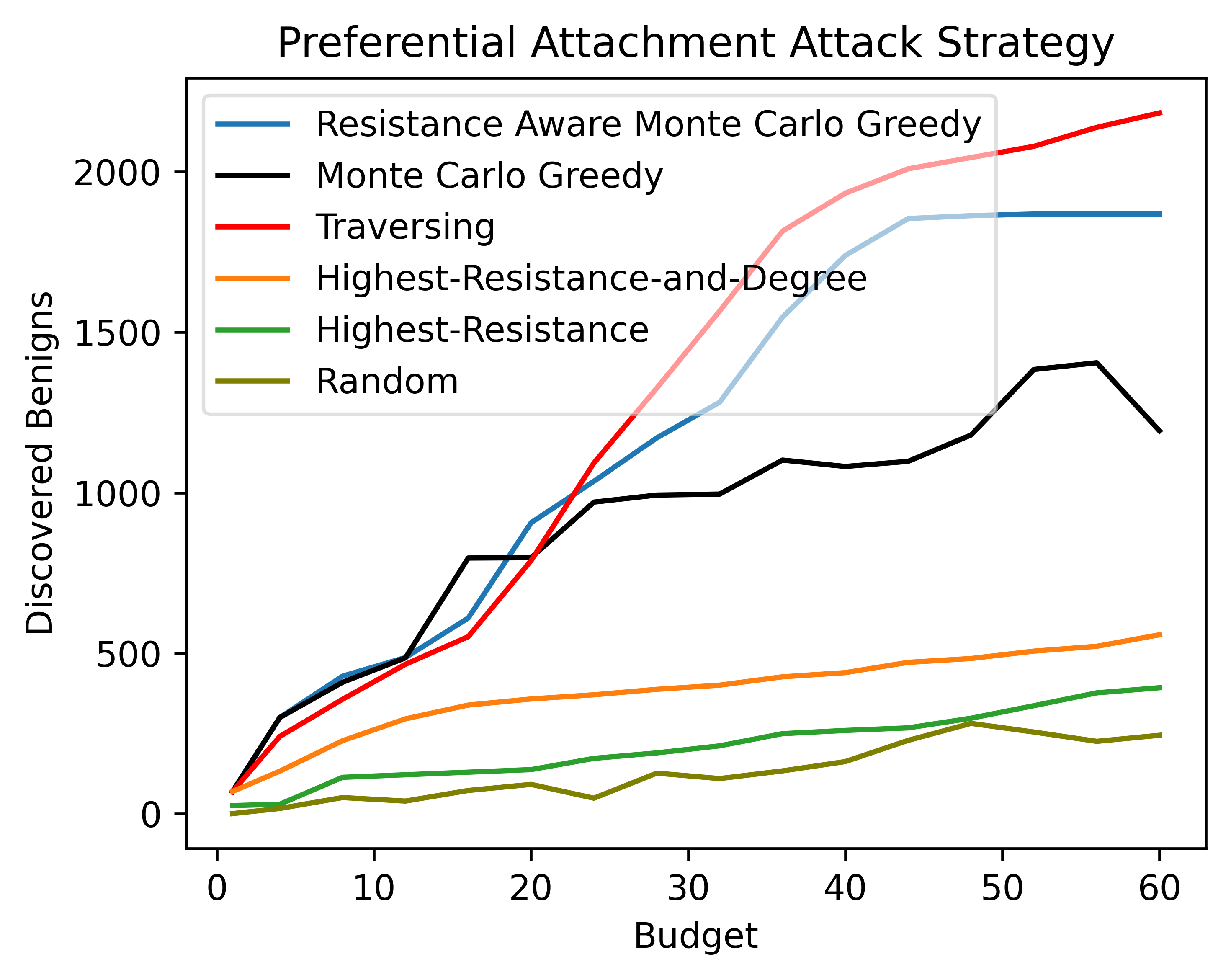}
    \includegraphics[width=0.32\linewidth]{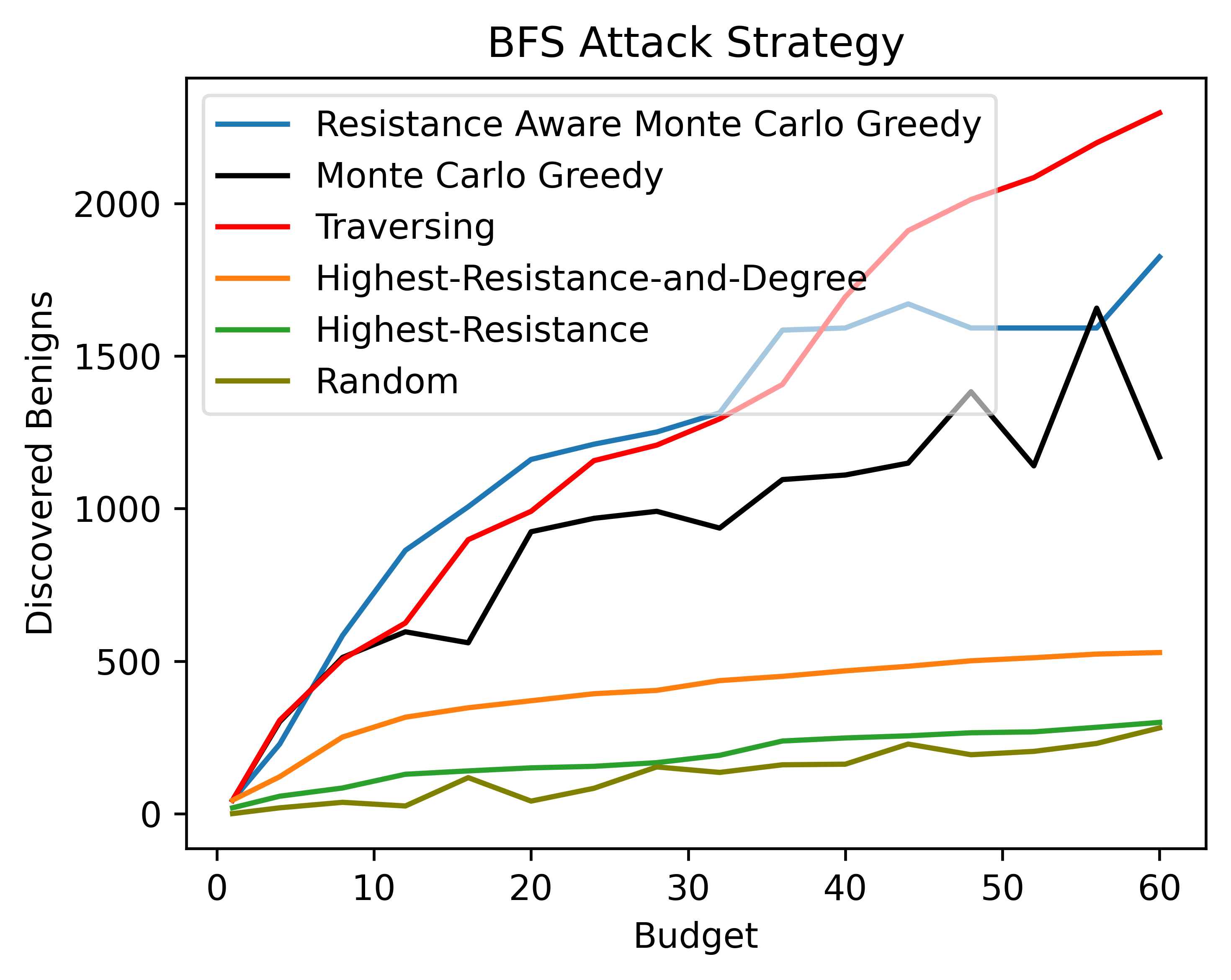}
    \caption{The number of discovered benigns by each algorithm when the budget ranges from $1$ to $k$ in the maximizing benigns problem on the \textit{LastFM} dataset and for different attack strategies. }
    \label{fig:mb_lastfm}

    \includegraphics[width=0.32\linewidth]{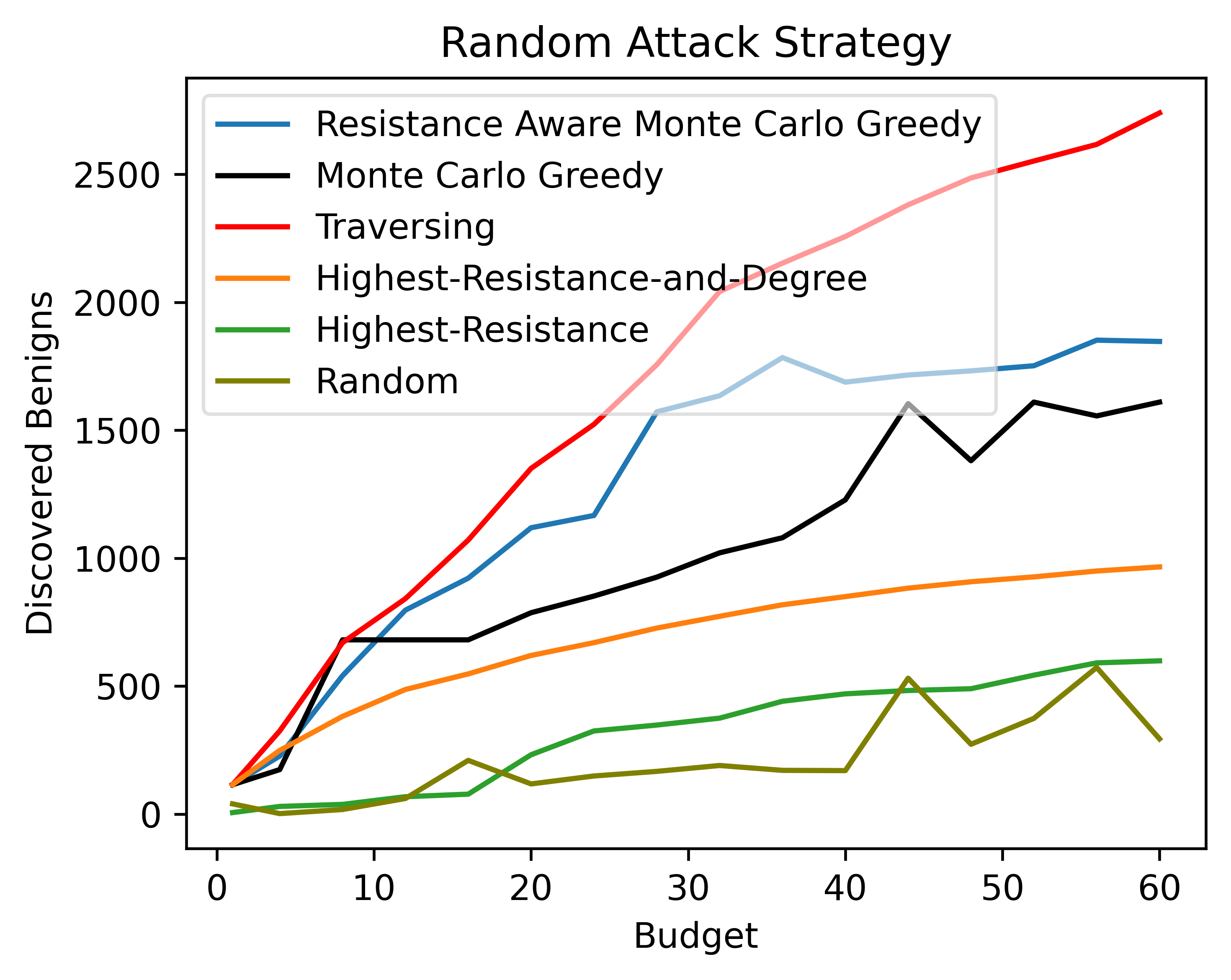}
    \includegraphics[width=0.32\linewidth]{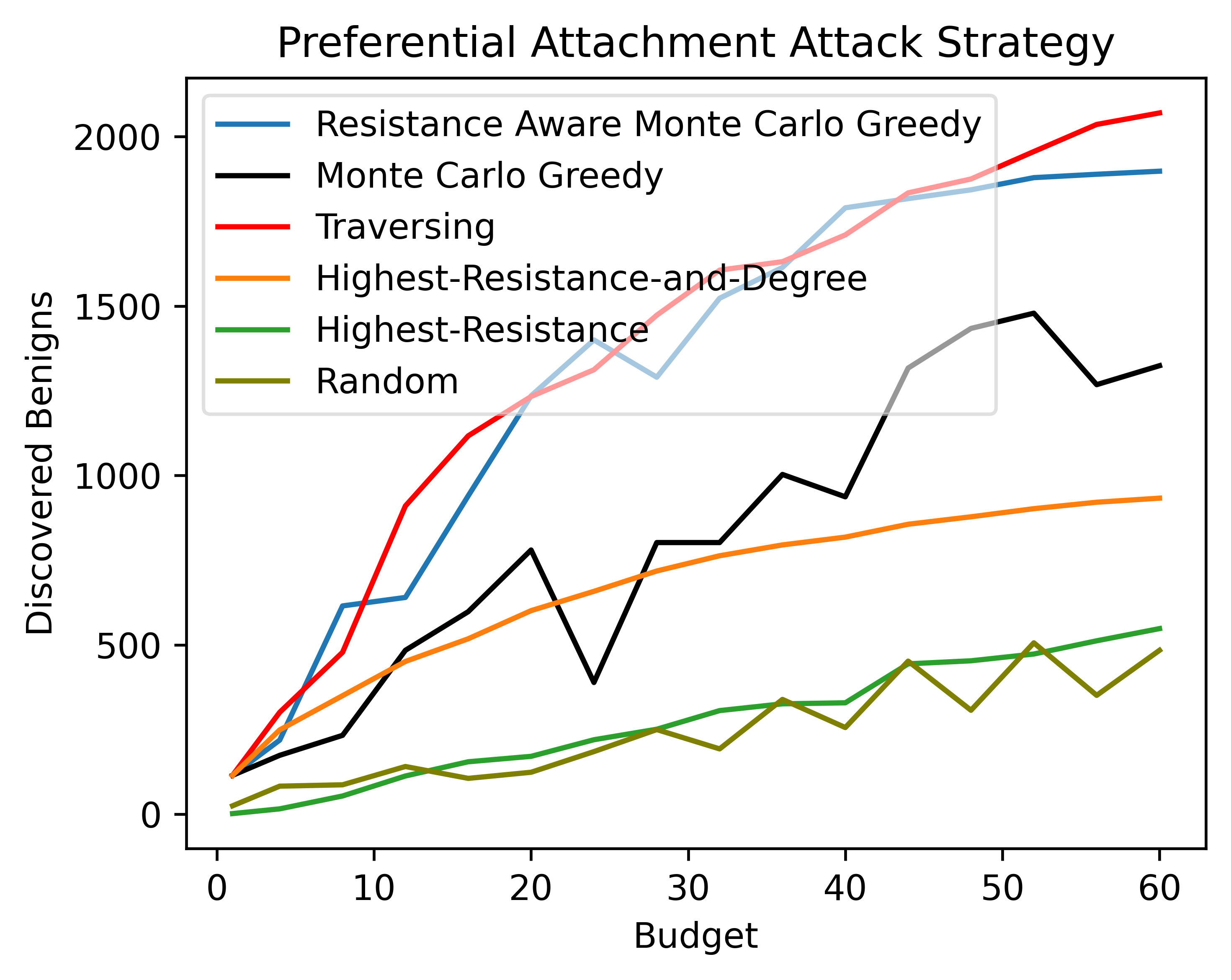}
    \includegraphics[width=0.32\linewidth]{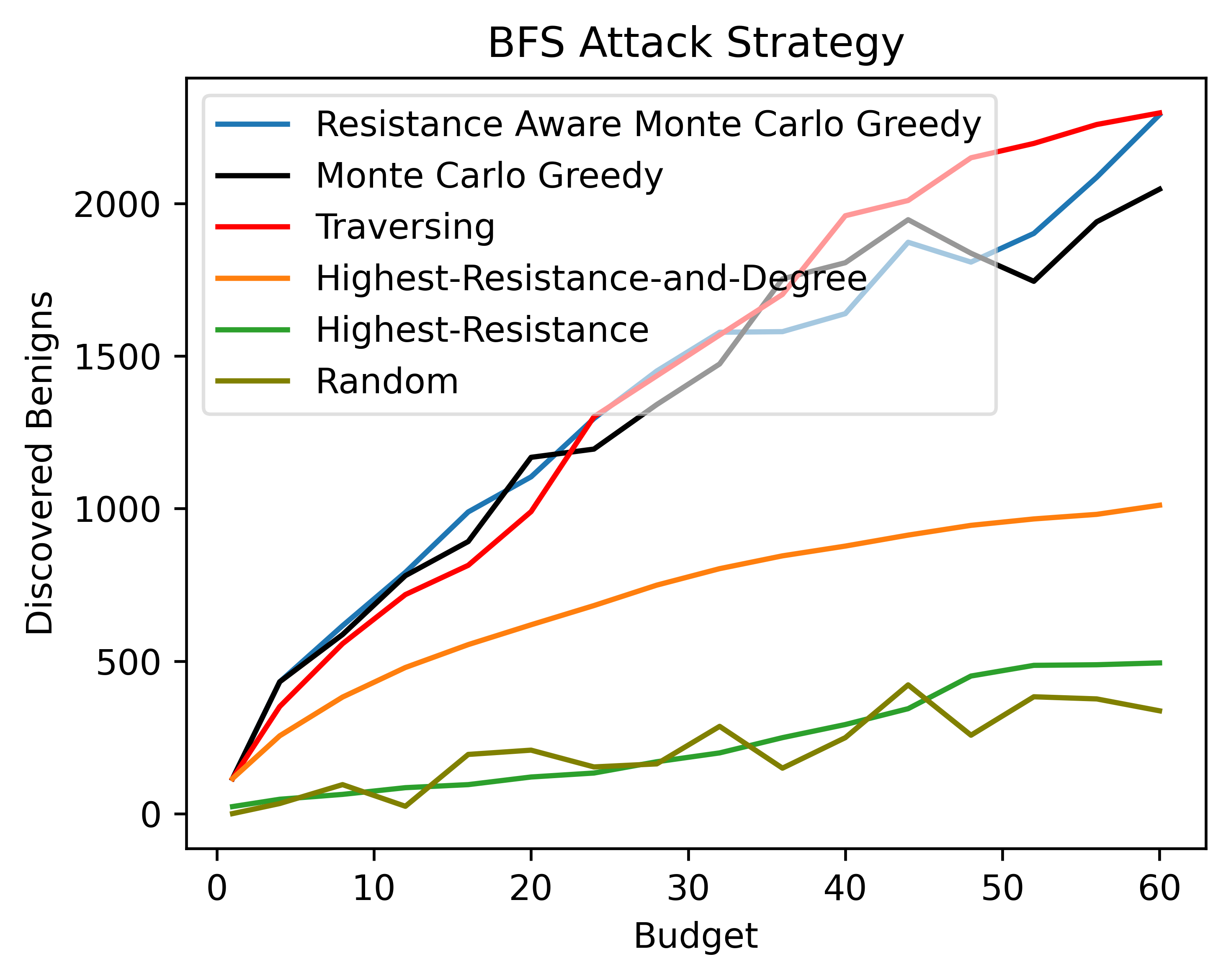}
    \caption{The number of discovered benigns by each algorithm when the budget ranges from $1$ to $k$ in the maximizing benigns problem on the \textit{Pokec} dataset and for different attack strategies. }
    \label{fig:mb_pokec}
    \Description{The figure illustrates the number of benigns discovered by different algorithms as the budget increases from $1$ to $k$ in the maximizing benigns problem across various datasets. The x-axis shows the budget size, and the y-axis shows the number of discovered benigns. Multiple plots represent different algorithms under various attack strategies, demonstrating the effectiveness of each method as the budget grows.}
\end{figure*}

\section{Additional Experiment for PAE}
\label{appendix:Additional PAE experiment reports}

For the experiments of discovering potential attack edges problem, the results on Twitter, LastFM, and Pokec datasets are provided in Figure~\ref{fig:PAE_all_datasets}.

\begin{figure*}[ht!]
    \centering
    \begin{subfigure}{0.49\textwidth}
        \centering
        \includegraphics[width=\linewidth]{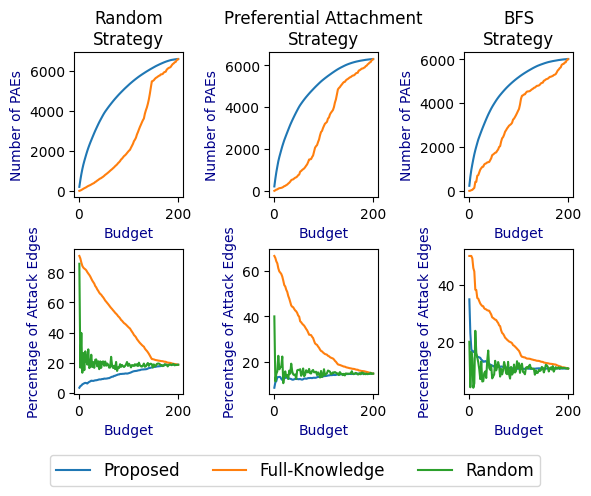}
        \caption{Twitter dataset}
        \label{fig:PAE_twitter}
    \end{subfigure}
    \hfill
    \begin{subfigure}{0.49\textwidth}
        \centering
        \includegraphics[width=\linewidth]{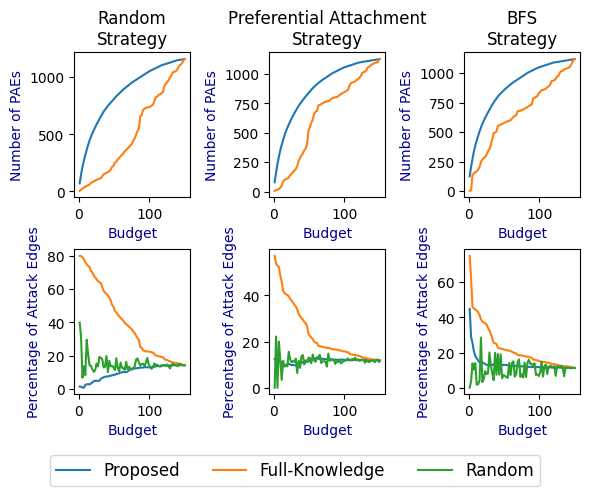}
        \caption{LastFM dataset}
        \label{fig:PAE_lastfm}
    \end{subfigure}
    
    \vspace{1em} 
    \begin{subfigure}{0.49\textwidth}
        \centering
        \includegraphics[width=\linewidth]{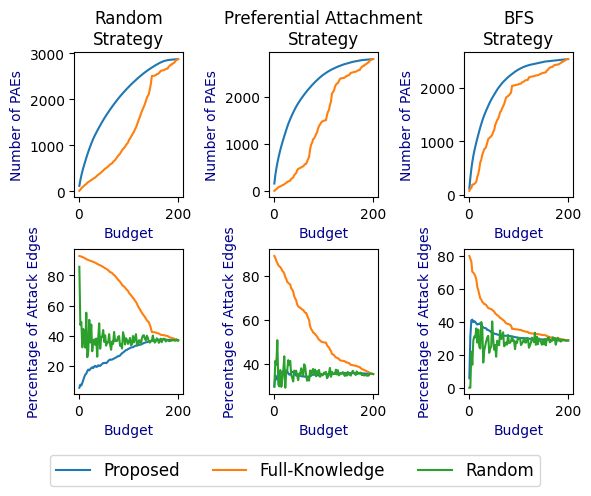}
        \caption{Pokec dataset}
        \label{fig:PAE_pokec}
    \end{subfigure}
    
    \caption{The number of discovered PAEs for different budgets (first row) and the percentage of attack edges relative to the number of discovered PAEs (second row) on \textit{Twitter} (a), \textit{LastFM} (b), and \textit{Pokec} (c) datasets. Each column represents a different attack strategy.
    }
    \Description{The figure has two rows and three columns of plots, each column representing a different attack strategy. The first row shows the number of discovered PAEs based on budget size, while the second row shows the percentage of attack edges relative to the discovered PAEs for each budget size.}
    \label{fig:PAE_all_datasets}
\end{figure*}

\section{Additional Experiments for Classification Algorithms}
\label{appendix:Other Experiments Classification}

For other node classification experiments on Twitter, LastFM and Pokec dataset, Table~\ref{table:combined_results} is provided.

\begin{table*}[ht!]
    \centering
    \caption{Performance of SybilSCAR, SybilWalk, and SybilMetric with and without preprocessing on the \textit{Twitter}, \textit{LastFM}, and \textit{Pokec} datasets. \textbf{Init} represents the setup without preprocessing. \textbf{MB} incorporates discovered benigns by the Traversing algorithm. \textbf{MB+PAE} incorporates both discovered benigns and PAEs.}
    \begin{tabular}{|c|c|c|c|c|c|}
        \hline
        \textbf{Dataset} & \makecell{\textbf{Attack} \\ \textbf{Strategy}} & \textbf{Step} & \makecell{\textbf{SybilSCAR} \\ \textbf{AUC}} & \makecell{\textbf{SybilWalk} \\ \textbf{AUC}} & \makecell{\textbf{SybilMetric} \\ \textbf{AUC}} \\
        \midrule
        \multirow{9}{*}{Twitter} & \multirow{3}{*}{Random} & Init & 0.901 & 0.918 & 0.82 \\
        \cline{3-6}
        & & MB & 0.935 & 0.972 & 0.95 \\
        \cline{3-6}
        & & MB+PAE & 0.936 & 0.973 & 0.92 \\
        \cline{2-6}
        & \multirow{3}{*}{BA} & Init & 0.891 & 0.889 & 0.83 \\
        \cline{3-6}
        & & MB & 0.934 & 0.961 & 0.96 \\
        \cline{3-6}
        & & MB+PAE & 0.936 & 0.966 & 0.96 \\
        \cline{2-6}
        & \multirow{3}{*}{BFS} & Init & 0.909 & 0.931 & 0.75 \\
        \cline{3-6}
        & & MB & 0.942 & 0.977 & 0.94 \\
        \cline{3-6}
        & & MB+PAE & 0.946 & 0.981 & 0.95 \\
        \hline \hline
        
        \multirow{9}{*}{LastFM} & \multirow{3}{*}{Random} & Init & 0.951 & 0.952 & 0.90 \\
        \cline{3-6}
        & & MB & 0.958 & 0.966 & 0.95 \\
        \cline{3-6}
        & & MB+PAE & 0.960 & 0.966 & 0.94 \\
        \cline{2-6}
        & \multirow{3}{*}{BA} & Init & 0.932 & 0.944 & 0.91 \\
        \cline{3-6}
        & & MB & 0.937 & 0.957 & 0.94 \\
        \cline{3-6}
        & & MB+PAE & 0.947 & 0.959 & 0.90 \\
        \cline{2-6}
        & \multirow{3}{*}{BFS} & Init & 0.948 & 0.971 & 0.88 \\
        \cline{3-6}
        & & MB & 0.953 & 0.980 & 0.95 \\
        \cline{3-6}
        & & MB+PAE & 0.961 & 0.981 & 0.92 \\
        \hline \hline
        
        \multirow{9}{*}{Pokec} & \multirow{3}{*}{Random} & Init & 0.895 & 0.924 & 0.90 \\
        \cline{3-6}
        & & MB & 0.916 & 0.946 & 0.94 \\
        \cline{3-6}
        & & MB+PAE & 0.913 & 0.948 & 0.89 \\
        \cline{2-6}
        & \multirow{3}{*}{BA} & Init & 0.852 & 0.916 & 0.90 \\
        \cline{3-6}
        & & MB & 0.893 & 0.936 & 0.94 \\
        \cline{3-6}
        & & MB+PAE & 0.901 & 0.940 & 0.93 \\
        \cline{2-6}
        & \multirow{3}{*}{BFS} & Init & 0.903 & 0.953 & 0.86 \\
        \cline{3-6}
        & & MB & 0.936 & 0.968 & 0.92 \\
        \cline{3-6}
        & & MB+PAE & 0.937 & 0.968 & 0.90 \\
        \hline
    \end{tabular}
    
    \label{table:combined_results}
\end{table*}
\clearpage

\end{document}